\documentclass[11pt]{article}
\usepackage{latexsym}
\usepackage{fullpage}
\usepackage{latexsym}
\usepackage{times}
\usepackage{epsfig}
\usepackage{amssymb}
\usepackage[ruled,vlined]{algorithm2e}

\newcommand{\be}{\begin{equation}}
\newcommand{\ee}{\end{equation}}
\newcommand{\bea}{\begin{eqnarray}}
\newcommand{\eea}{\end{eqnarray}}
\newcommand\assign{\leftarrow}

\newcommand{\varoptk}{\textnormal{\sc VarOpt$_k$}}
\newcommand{\uniformk}{\textnormal{\sc Unif$_k$}}
\newcommand{\uniformkbase}{\textnormal{\sc Unif$_{k,k+1}$}}
\newcommand{\varoptkbase}{\textnormal{\sc VarOpt$_{k,k+1}$}}
\newcommand{\var}{\mathop{\sf Var}}
\newcommand{\cov}{\mathop{\mathsf{Cov}}}
\newtheorem{theorem}{Theorem} 
\newtheorem{lemma}[theorem]{Lemma} 
\newtheorem{proposition}[theorem]{Proposition} 
 
\newcommand{\qed}{\hfill \rule{1ex}{1ex}\medskip\\} 
\newenvironment{proof}[1][]{\paragraph{Proof{#1}}}{\qed}

\newcommand{\drop}[1]{}

\def\varopt{\textnormal{\sc VarOpt}}
\def\pri{\mbox{\sc pri}}
\def\ws{\mbox{\sc ws}}
\def\rc{\mbox{\sc RC}}
\def\ssc{\mbox{\sc SC}}

\let\hat\widehat
\def\Rl{{\mathbb R}}

\def\P{\textsf{P}}
\newcommand\E{\textsf{E}}

\let\Var\var

\newcommand\Prp[1]{\Pr\!\left[{{#1}}\right]}

\newcommand\req[1]{(\ref{#1})}

\newcommand\eps\varepsilon

\def\ol#1{\overline{#1}}

\newcommand{\VS}{{V\Sigma}}
\newcommand{\SV}{{\Sigma}V}
\newcommand{\SCV}{{\Sigma}CoV}
\begin{document}

\title{Stream sampling for variance-optimal estimation of subset sums\thanks{An extended abstract of this paper was presented at the 20th ACM-SIAM Symposium on Discrete Algorithms, 2009.}}

\author{Edith Cohen\thanks{AT$\&$T Labs---Research, Florham Park, NJ, USA (email: \texttt{(edith,duffield,lund,mthorup)@research.att.com})} \and Nick Duffield$^\dagger$ \and Haim Kaplan\thanks{The Blavatnik School of Computer Science, Tel Aviv University, Israel (email: \texttt{haimk@cs.tau.ac.il})} \and Carsten Lund$^\dagger$ \and Mikkel Thorup$^\dagger$ }

\date{}
\maketitle

\date{}

\begin{abstract}%\small\baselineskip=9pt 
From a high volume stream of weighted items, we want to maintain a
generic sample of a certain limited size $k$ that we can later use to
estimate the total weight of arbitrary subsets. This is the classic
context of on-line reservoir sampling, thinking of the generic sample
as a reservoir. We present an efficient reservoir sampling scheme, $\varoptk$,
that dominates all previous schemes in terms of estimation quality.
  $\varoptk$ provides {\em variance optimal unbiased
estimation of subset sums}. More precisely, if we have seen $n$ items
of the stream, then for {\em any} subset size $m$, our scheme based on $k$
samples minimizes the average variance over all subsets of size
$m$. 
In fact, the optimality is against any off-line scheme with $k$
samples tailored for the concrete set of items seen. In addition to optimal average
variance, our scheme provides tighter worst-case bounds on the
variance of {\em particular} subsets than previously possible.
It is efficient, handling each new item of
the stream in $O(\log k)$ time.
Finally, it is particularly well suited for combination of samples
from different streams in a distributed setting.
\end{abstract}
\section{Introduction}
In this paper we focus on sampling from a high volume stream of
weighted items.  The items arrive 
faster and in larger quantities than
can be saved, so only a sample can be stored efficiently.  We want to
maintain a generic sample of a certain limited size that we can later
use to estimate the total weight of {\em arbitrary} subsets.

This is a fundamental and practical problem. In \cite{JMR05} this is
the basic function used in a database system for streams. Such a
sampling function is now integrated in a measurement system for Internet traffic analysis \cite{Gigascope:sigmod03}.
In this context, items are records
summarizing the flows of IP packets streaming by a router.  Queries on
selected subsets have 
numerous current and potential applications,
including anomaly detection (detecting unusual traffic patterns by
comparing to historic data), traffic engineering and routing 
(e.g., estimating traffic
volume between Autonomous System (AS) pairs), and billing (estimating volume of traffic
to or from a certain source or destination).
It is important that we are not constrained to subsets
known in advance of the measurements. This would preclude
exploratory studies, and would not allow a change in routine questions
to be applied retroactively to the measurements. 
A striking example where the selection is not known in advance was the
tracing of the {\em Internet Slammer Worm} \cite{MPSSSW03}. It turned
out to have a simple signature in the flow record; namely as being udp
traffic to port 1434 with a packet size of 404 bytes.  Once this
signature was identified, the worm could be studied by selecting
records of flows matching this signature from the sampled
flow records. 

We introduce a new sampling and estimation scheme for streams, denoted
$\varoptk$, which selects $k$ samples from $n$ items. $\varoptk$ has
several important qualities: All estimates are {\em unbiased}. The
scheme is {\em variance optimal} in that it simultaneously minimizes
the average variance of weight estimates over subsets of
\textsl{every} size $m<n$.  The average variance optimality is
complemented by optimal {\em worst-case} bounds limiting the variance
over all combinations of input streams and queried subsets.  These
per-subset worst-case bounds are critical for applications requiring
robustness and for the derivation of confidence intervals.
%Mikkel-10/1: A lot of changes from here and till Section 1.1
Furthermore, $\varoptk$ is fast. It handles each item in $O(\log k)$
worst-case time, and $O(1)$ expected amortized time for randomly
permuted streams. 

In Section \ref{exp:sec} (Figure \ref{netflix:fig}) we demonstrate the
estimation quality of $\varoptk$ experimentally via a comparison with other
reservoir sampling schemes on the Netflix Prize data set
\cite{netflix}. With our implementation of $\varoptk$, the time to
sample 1,000 items from a stream of 10,000,000 items was only 7\%
slower than the time required to read them.

Ignoring the on-line efficiency for streams, there has been several
schemes proposed that satisfy the above variance properties both from
statistics \cite{Cha82,Til96} and indirectly from computer science
\cite{Sri01}.  Here we formulate the sampling operation $\varoptk$ as
a general recurrence, allowing independent $\varoptk$ samples from
different subsets to be naturally combined to obtain a $\varoptk$
sample of the entire set. The schemes from \cite{Cha82,Til96} fall out
as special cases, and we get the flexibility needed for fast on-line
reservoir sampling from a stream. The nature of the recurrence is also
perfectly suited for distributed settings.

Below we define the above qualities more precisely and present an
elaborate overview of previous work.

\subsection{Reservoir sampling with unbiased estimation}
The problem we consider is classically known as reservoir sampling
\cite[pp. 138--140]{Knu69}.  In reservoir sampling, we process a
stream of (weighted) items. The items arrive one at the time, and a
reservoir maintains a sample $S$ of the items seen thus far. When a
new item arrives, it may be included in the sample $S$ and old
items may be dropped from $S$. Old items outside $S$ are never
reconsidered.   We think of estimation as an integral part of
sampling. Ultimately, we want to use a sample to estimate the total
weight of any subset of the items seen so far. Fixing notation, we
are dealing with a stream of items where item $i$ has a positive
weight $w_i$. For some integer capacity $k\geq 1$, we maintain a
reservoir $S$ with capacity for at most $k$ samples from the items
seen thus far.  Let $[n]=\{1,\ldots ,n\}$ be the set of items
seen. With each item $i\in S$ we store a weight estimate $\hat
w_i$, which we also refer to as {\em adjusted weight}. 
For items $i\in [n]\setminus S$  we have an
implicit zero estimate $\hat w_i=0$.  We require these estimators to
be unbiased in the sense that $\E[\hat w_i]=w_i$.  A typical example
is the classic Horvitz-Thompson estimator \cite{HT52} setting $\hat
w_i=w_i/\Pr[i\in S]$ if $i\in S$. 

Our purpose is to estimate arbitrary subset sums from the sample. 
For any subset
$I\subseteq [n]$, we let $w_I$ and $\hat w_I$ denote $\sum_{i\in
I}w_i$ and $\sum_{i\in I}\hat w_i$, respectively.  By linearity of
expectation $\E[\hat w_I]=w_I$. Since all unsampled items have
$0$ estimates, we get $\hat w_{I\cap S}=\hat w_I$. Thus
$\hat w_{I\cap S}$, the sum of the adjusted weights of items
from the sample that are members of $I$, 
is an unbiased estimator of $w_I$.

Reservoir sampling thus addresses two
issues:
\begin{itemize}
\item The streaming issue \cite{Mut05} where with limited
memory we want to compute a sample from a huge stream that passes by only once. 
\item The incremental data structure issue of maintaining a sample as new weighted items
are inserted. In our case, we use the sample to provide quick
estimates of sums over arbitrary subsets of the items seen thus far.
\end{itemize}
Reservoir versions of different sampling schemes are presented 
in~\cite{CMN99,CK:podc07,DLT07,FMR62,ES:IPL2006,Vit85}.

\drop{
$\varoptk$ is
differentiated from previous schemes by a particularly %elegant and
%Mikkel: I removed elegant because priority sampling is arguably
%far more elegant, at least to implement, but clean is a good term.
clean design: Other schemes derive the adjusted weights, including the
final ones, using auxiliary information, e.g., a the total weight seen
thus far.  $\varoptk$ maintains adjusted weights as part of the
on-line process without  {\em any\/} additional auxiliary
information.  More precisely, $\varoptk$ does not distinguish between
previously-sampled items with adjusted weights and a
new item with its original weights.  $\varoptk$ processes a new item
from the data stream by applying variance-optimal sampling of $k$ out
of $k+1$ items to the current reservoir and the new item
%Mikkel6/25
without distinction. Surprisingly, it turns out that each such
``local'' step is uniquely determined by variance-optimality and that
global variance optimality follows.
}

\subsection{Off-line sampling}
When considering the qualities of the sample, we compare our on-line 
scheme, $\varoptk$, with a powerful arbitrary off-line sampling scheme
which gets the $n$ weighted items up front, and can tailor the sampling
and estimation freely to this concrete set, not having to worry about
efficiency or the arrival of more items. The only restriction is
the bound $k$ on the number of samples. More abstractly, the off-line
sampling scheme is an arbitrary probability distribution $\Omega$ over 
functions $\hat w:[n]\rightarrow \mathbb R$ from items $i$ to weight 
estimates $\hat w_i$ which is unbiased
in the sense that $\E_{\hat w\leftarrow \Omega}[\hat w_i]=w_i$, and which has at most
$k$ non-zeros.

\subsection{Statistical properties of target}\label{sec:properties}
%Mikkel-10/1: changed mathematical to statistical in target as
%math also covers running time.
The sampling scheme we want should satisfy some classic goals from statistics.
Below we describe these goals. Later we will discuss their relevance to
subset sum estimation. 

\paragraph{(i)} {\em Inclusion probabilities proportional to size (ipps)}.
To get
$k$ samples, we want each item $i$ to be sampled with probability
$p_i=kw_i/w_{[n]}$.  This is not possible if some item $j$ has more
than a fraction $k$ of the total weight. In that case, the standard is
that we include $j$ with probability $p_j=1$, and recursively ipps
sample $k-1$ of the remaining items. In the special case where we start
with $k\geq n$, we end up including
all items in the sample. The included items are given the
standard Horvitz-Thompson estimate $\hat w_i=1/p_i$.

Note that ipps only considers the marginal distribution on each item, so many
joint distributions are possible and in itself, it only leads
to an expected number of $k$ items.\smallskip

\paragraph{(ii)} {\em Sample contains at most $k$ items}. Note that (i) and (ii) together
implies that the sample contains exactly $\min\{k,n\}$ items.\smallskip

\paragraph{(iii)} {\em No positive covariances} between distinct
adjusted weights.\smallskip

%Mikkel-10/1
\noindent From statistics, we know several schemes satisfying the
above goals (see, e.g., \cite{Cha82,Til96}), but they
are not efficient for on-line reservoir sampling. In addition to the above
goals, we will show that $\varoptk$ estimates admit standard 
Chernoff bounds.

\subsection{Average variance optimality}\label{sec:average}
Below we will discuss some average variance measures that are automatically 
optimized by goal (i) and (ii) above.

When $n$ items have arrived, for each subset size $m\leq n$, 
we consider the average variance for subsets of size $m\leq n$:
\[V_m=\E_{I\subseteq [n], |I|=m}\left[\Var[\hat w_I]\right]
=\frac{\sum_{I\subseteq [n], |I|=m}\left[\Var[\hat w_I]\right]}{{n
\choose m}}.\] 
Our $\varoptk$ scheme is variance optimal in the following strong sense.
For each reservoir size $k$,
stream prefix of $n$ weighted items, and subset size $m$,
there is no off-line sampling scheme with $k$ samples
getting a smaller average variance $V_m$ than our generic $\varoptk$.

The average variance measure $V_{m}$ was introduced in \cite{ST07} where it
was proved that
\begin{equation}\label{eq:Vm}
V_{m}=\frac mn\left(\frac{n-m}{n-1}\,\SV+\frac{m-1}{n-1}\,\VS\right) \ ,
\end{equation}
Here $\SV$ is the sum of individual variances while $\VS$ is the
variance of the estimate of the total, that is,
\begin{eqnarray*}
\SV&=&\sum_{i\in[n]}\Var[\hat w_i]\ =\ n V_1\\
 \VS&=&\Var[\sum_{i\in[n]}\hat w_i]\ =\ \Var[\hat w_{[n]}]\ =\ V_n.
\end{eqnarray*}
It follows that we minimize $V_m$ for all $m$ if and only if we
simultaneously minimize $\SV$ and $\VS$, which is exactly what $\varoptk$ does.
The optimal value for $\VS$ is $0$, meaning that the estimate of
the total is exact.

Let $W_p$ denote  the expected variance
of a random subset including each item $i$ independently with some
probability $p$. It is also shown in \cite{ST07} that
$W_{p}=p\left((1-p)\SV+p\VS\right)$. 
So if we 
simultaneously minimize $\SV$ and $\VS$, we also minimize 
$W_{p}$. 
%Mikkel6/27
It should be noted that both $\SV$ and $\VS$ are known measures 
from statistics (see, e.g., \cite{SSW92} and concrete examples in the
next section). It is the implications for average variance over 
subsets that are from \cite{ST07}.

With no information given about which kind of subsets are to be estimated, it
makes most sense to optimize average variance measures like those above
giving each item equal opportunity
to be included in the estimated subset. If the
input distributions are not too special, then we expect
this to give us the best estimates in practice, using variance as
the classic measure for estimation quality.

\paragraph{Related auxiliary variables}\label{sec:aux}
We now consider the case where we for each item are interested in an auxiliary weight $w'_i$.
For these we use the estimate 
${\hat w}'_i=w'_i\hat w_i/w_i$,
which is unbiased since $\hat w_i$ is unbiased. 
Let $\VS'=\sum_{i\in[n]}{\hat w}'_i$ 
be the variance on the estimate of the total for the auxiliary variables. 

We will argue that we expect to do best possible on $\VS'$ using $\varoptk$,
assuming that the $w'_i$ are randomly generated from the $w_i$. Formally
we assume each $w'_i$ is generated as $w'_i=x_iw_i$
where the $x_i$ are drawn independently from the same distribution $\Xi$. 
We consider
expectations $\E_\Xi$ for random choices of the vector $x=(x_i)_{i\in[n]}$, that is, formally $\E_\Xi[\VS']=\E_{x\leftarrow \Xi^n}\left[\VS'\,|\,x\right]$.
We will prove
\begin{equation}\label{eq:aux}
\E_\Xi[\VS']=\var[\Xi]\SV+\E[\Xi]^2\VS\textnormal,
\end{equation}
where $\var[\Xi]=\var_\Xi[x_i]$ and $\E[\Xi]=\E_\Xi[x_i]$ for every $x_i$.
From \req{eq:aux} it follows that we minimize $\E_\Xi[\VS']$ when
we simultaneously minimize $\SV$ and $\VS$ as we do with $\varoptk$.
Note that if the $x_i$ are
0/1 variables, then the ${\hat w}'_i$ represent a random subset, 
including each item independently as in $W_p$ above. The proof
of \req{eq:aux} is found in Appendix \ref{sec:aux-proof}

\paragraph{Relation to statistics}
The above auxiliary variables can be thought of as modeling a
classic scenario in statistics, found in text books such as
\cite{SSW92}. We are interested in some weights $w'_i$ that will only
be revealed for sampled items. However, for every $i$, we have
a known approximation $w_i$ that we can use in deciding which items
to sample. As an example, the $w_i'$ could be household incomes while
the $w_i'$ where approximations based on postal codes. The main purpose
of the sampling is to estimate the total of the $w'_i$. 
When evaluating different schemes, \cite{SSW92} considers $\VS$, stating that
if the $w'_i$ are proportional to the $w_i$, then the variance $\VS'$ on
the estimated total $\sum_i\hat w_i'$ is proportional
to $\VS$, and therefore we should minimize $\VS$. This corresponds
to the case where $\var[\Xi]=0$ in \req{eq:aux}. However, \req{eq:aux}
shows that $\SV$ is also important to $\VS'$ if the relation
between $w_i'$ and $w_i$ is not just proportional, but also has a
random component.

As stated, $\SV$ is not normally the focus in statistics, but
for Poisson sampling where each item is sampled independently,
we have $\SV=\VS$, and studying this case, it is shown in  \cite[p. 86]{SSW92} 
that the ipps of goal (i) uniquely minimizes $\SV$ (see \cite{DLT07} for a proof
working directly on the general case allowing dominant items). 
It is also easy to verify that
conditioned on (i), goal (ii) is equivalent to $\VS=0$ (again this
appears to be standard, but we couldn't find a reference for the 
general statement.
The argument is trivial though. Given the (i), the only variability 
in the weight estimates returned is in the number of sampled estimates
of value $\tau$, so the estimate of the total is variable if and only
if the number of samples is variable). The classic goals (i) and (ii) are thus equivalent to minimizing $\SV$ and $\VS$, hence all the
average variances discussed above. 

\subsection{Worst-case robustness}\label{sec:robustness}
In addition to minimizing the average variance, $\varoptk$ has some
complimentary worst-case robustness properties, limiting the variance
for every single (arbitrary) subset. We note that any such bound has
to grow with the square of a scaling of the weights.  This kind of robustness
is important for applications seeking to minimize worst-case
vulnerability. The robustness discussed below is all a consequence of
the ipps of goal (i) combined with the non-positive covariances of goal (iii).

With the Horvitz-Thompson estimate, the variance of item $i$ is
$w_i^2(1/p_i-1)$. With ipps sampling, $p_i\geq\min\{1,k w_i/w_{[n]}\}$. This
gives us the two bounds $\Var[\hat w_i]< w_iw_{[n]}/k$ and 
$\Var[\hat w_i]< (w_{[n]}/(2k))^2$ (for the second bound
note that $p_i<1$ implies $w_i<w_{[n]}/k$). Both of these bounds are asymptotically
tight in that sense that there are instances for which no sampling scheme
can get a better leading
constant. More precisely, the bound $\Var[\hat w_i]< w_iw_{[n]}/k$ is asymptotically tight 
if every $i$ has $w_i=o(w_{[n]}/k)$, e.g., when
sampling $k$ out of $n$ units, the individual
variance we get is $(n/k)-1$. The bound $(w_{[n]}/(2k))^2$ is tight
for $n=2k$ unit items.
In combination with the non-positive covariances of goal (iii), we 
get that every
subset $I$ has weight-bounded variance $\Var[\hat w_I]\leq w_I w_{[n]}/k$,
and cardinality bounded variance $\Var[\hat w_I]\leq |I|(w_{[n]}/2k)^2$.

\drop{
\subsection{Confidence bounds}\label{sec:confidence}
Property (iii$+$) allows us to apply Chernoff-like derivations~\cite{PancSri:sicomp97,Sri01},
that are based on the assumption that items' inclusions in
the sample are independent random variables, to $\varoptk$ samples, and
obtain upper and lower bounds on subset sizes.
}

\subsection{Efficient for each item}
With $\varoptk$ we can handle each new item of the stream in $O(\log
k)$ worst-case time. 
In a realistic implementation with floating point
numbers, we have some precision $\wp$ and accept an error of
$2^{-\wp}$. We will prove an $\Omega(\log k/\log\log k)$ lower bound on the
worst-case time for processing an item on the word RAM for any
floating point implementation of a reservoir sampling scheme with
capacity for $k$ samples which satisfies goal (i) minimizing $\SV$. 
Complementing that we
will show that it is possible to handle each item in $O(\log\log k)$
amortized time. If the stream is viewed as a random permutation of the 
items, we will show that the expected amortized cost per item is only constant.

\subsection{Known sampling schemes}\label{sec:all-schemes}
We will now discuss known sampling schemes in relation to the qualities
of our new proposed scheme:
\begin{itemize}
\item Average variance optimality of Section \ref{sec:average} following
from goal (i) and (ii).
\item The robustness of Section \ref{sec:robustness} following from
goal (i) and (iii).  
%Mikkel-10/1: Applicability of confidence bounds, following from (iii+).
\item Efficient reservoir sampling implementation with capacity
for at most $k$ samples; efficient distributed implementation.
\end{itemize}
The statistics literature contains many sampling
schemes~\cite{SSW92,Tille:book}
that share some of these qualities, but then they all
perform significantly worse on others. 

 % For example, many sampling
 % schemes from statistics, such as Sunter's method \cite{Sun77}, are not
 % suitable in our reservoir sampling context because they need to sort the items
 % before deciding which ones to sample. 
 % We note that the statistics literature contains many sampling schemes
 % with many different features, and here we can only discuss a select
 % few schemes that we consider most relevant in our context. 

 % In our discussion, we are
% particularly interested in heavy-tailed distributions (such as
% power-low distributions), that frequently occur in practice.  In
% these distributions, a small fraction of dominant items accounts for
% a large fraction of the total weight \cite{AFT98,PKC96}.  

\paragraph*{Uniform sampling without replacement}
In uniform sampling without replacement, we pick a sample of $k$ items
uniformly at random. If item $i$ is sampled it gets the
Horvitz-Thompson weight estimate $\hat w_i=w_i n/k$.  Uniform sampling
has obvious variance problems with the frequently-occurring heavy-tailed
power-low distributions, where
a small fraction of dominant items accounts for
a large fraction of the total weight \cite{AFT98,PKC96},
 because it is likely to miss the dominant items. 

\paragraph*{Probability proportional to size sampling with replacement (ppswr)}
In probability proportional to size sampling (pps) with replacement (wr), each sample $S_j\in [n]$, $j\in[k]$, is
independent, and equal to $i$ with probability $w_i/w_{[n]}$. Then
$i$ is sampled if $i=S_j$ for some $j\in[k]$. This happens with
probability $p_i=1-(1-w_i/w_{[n]})^k$, and if $i$ is sampled, it
gets the Horvitz-Thompson estimator $\hat w_i=w_i/p_i$. Other
estimators have been proposed, but we always have the same problem
with heavy-tailed distributions: if a few dominant items contain most of
the total weight, then most samples will be copies of these dominant
items. As a result, we are left with comparatively few samples of
the remaining items, and few samples imply high variance no matter
which estimates we assign.

\paragraph*{Probability proportional to size sampling without replacement (ppswor)}
An obvious improvement to ppswr is to sample without replacement
(ppswor). Each new item is then chosen with probability proportional
to size among the items not yet in the sample. 
With ppswor, unlike ppswr,
the probability that an item is included in the sample is a 
complicated function of all
the item weights, and therefore the Horvitz-Thompson
estimator is not directly applicable. A ppswor 
reservoir sampling and estimation procedure is, however, 
presented in \cite{CK:podc07,CK07,bottomk:VLDB2008}. 

Even though ppswor resolves the ``duplicates problem'' of ppswr, we claim
here a negative result for {\em any} ppswor estimator: 
in Appendix \ref{sec:ppswor-bad}, we will present an
instance for any sample size $k$ and number of items $n$  such
that any estimation based on up to $k+(\ln k)/2$ ppswor samples will 
perform a factor $\Omega(\log k)$ worse than $\varoptk$
for {\em every} subset size $m$. This 
is the first such negative result for the classic ppswor besides
the fact that it is not strictly optimal. 

\paragraph*{Ipps Poisson sampling} 
It is more convenient to think of ipps sampling in terms of a {\em threshold\/}
$\tau$. We include in the sample $S$ every item with weight $w_i\geq\tau$, using
the original weight as estimate $\hat w_i=w_i$.
An item $i$ with weight $w_i<\tau$ is included with probability $p_i=w_i/\tau$,
and it gets weight estimate $\tau$ if sampled.

For an expected number of $k<n$ samples, we use the unique
$\tau=\tau_k$ satisfying 
\begin{equation}\label{eq:thr}
\sum_i p_i =\sum_i\min\{1,w_i/\tau_k\}=k.
\end{equation}
For $k\geq n$, we define
$\tau_k=0$ which implies that all items are included. This threshold centric view of
ipps sampling is taken from \cite{DLT05}.

% Recall that the ipps sampling of goal (i) is defined in terms of a threshold $\tau$:
% we include every item with weight $w_i\geq\tau$, using
% the original weight as estimate $\hat w_i=w_i$.
% An item $i$ with weight $w_i<\tau$ is included with probability $p_i=w_i/\tau$,
% and it gets weight estimate $\tau$ if sampled. 

If the threshold $\tau$ is given, and if we are satisfied with Poisson sampling,
that is, each item is sampled independently, then we can trivially
perform the sampling from a stream. In \cite{DLT07} it is shown how we can
adjust the threshold as samples arrive to that we always have a reservoir with
an expected number of $k$ samples, satisfying goal (i) for the items seen
thus far. Note, however, that we may easily violate goal (ii) of having
at most $k$ samples.

Since the items are sampled independently, we have zero covariances, so
(iii) is satisfied along with the all the robustness of Section
\ref{sec:robustness}. However, the average variance of Section
\ref{sec:average} suffers. More precisely, with zero covariances, we
get $\VS=\SV$ instead of $\VS=0$. From \req{eq:Vm} we get that for
subsets of size $m$, the average variance is a factor $(n-1)/(n-m)$
larger than for a scheme satisfying both (i) and (ii).  Similarly we
get that the average variance $W_{\frac12}$ over all subsets is larger
by a factor 2.

\paragraph*{Priority sampling}
Priority sampling was introduced in \cite{DLT07} as a threshold
centric scheme which is tailored for reservoir sampling with $k$ as a
hard capacity constraint as in (ii). It is proved in \cite{Sze06} that
priority sampling with $k+1$ samples gets as good $\SV$ as the optimum
obtained by (i) with only $k$ samples. Priority sampling has zero
covariances like the above ipps Poisson sampling, so it satisfies
(iii), but with $\VS=\SV$ it has the same large average variance for
larger subsets.

\paragraph*{Satisfying the goals but not with efficient reservoir sampling}
As noted previously, there are several schemes satisfying all our
goals \cite{Cha82,Til96,Sri01}, but they are not efficient for
reservoir sampling or distributed data.  Chao's scheme \cite{Cha82}
can be seen as a reservoir sampling scheme, but when a new item
arrives, it computes all the ipps probabilities from scratch in $O(n)$
time, leading to $O(n^2)$ total time. Till\'e \cite{Til96} has
off-line scheme that eliminates items from possibly being in the
sample one by one (Till\'e also considers a complementary scheme that
draws the samples one by one).  Each elimination step involves
computing elimination probabilities for each remaining item. As such,
he ends up spending $O((n-k)n)$ time ($O(kn)$ for the complementary
scheme) on selecting $k$ samples. Srinivasan \cite{Sri01} has
presented the most efficient off-line scheme, but cast for a different
problem. His input are the desired inclusion probabilities $p_i$ that
should sum to $k$.  He then selects the $k$ samples in linear time by
a simple pairing procedure that can even be used on-line. However, to
apply his algorithm to our problem, we first need to compute the ipps
probabilities $p_i$, and to do that, we first need to know all the
weights $w_i$, turning the whole thing into an off-line linear time
algorithm. Srinivasan states that he is not aware of any previous
scheme that can solve his task, but using
his inclusion probabilities, the above mentioned older schemes
from statistics \cite{Cha82,Til96} will do the job, albeit less
efficiently. 
  % Both his and Till\'e's procedure are inherently off-line in
  % that they need to know all the inclusion probabilities before they
  % can start sampling. 
We shall discuss our technical relation to \cite{Cha82,Til96}
in more detail in Section \ref{sec:Chao-Tille}.
Our contribution is a scheme $\varoptk$ that satisfies all our goals
(i)--(iii) while being efficient reservoir sampling from a stream,
processing each new item in $O(\log k)$ time.

\subsection{Contents}\label{sec:contents}
In Section \ref{sec:varopt} we will present our recurrence to generate 
$\varoptk$ schemes, including those from \cite{Cha82,Til96} as special cases. 
In Section \ref{sec:recur} we will prove that the general method works.
In Section \ref{sec:efficient} we will present efficient implementations,
complemented in Section \ref{sec:lower-bound} with a lower bound.
In Section \ref{exp:sec} we present an experiment comparison with
other sampling and estimation scheme. Finally, in Section \ref{sec:Chernoff}
we prove that our $\varoptk$ schemes actually admit the kind of Chernoff
bounds we usually associate with independent Poisson samples.

\section{\varoptk}\label{sec:varopt}
By $\varoptk$ we will refer to any unbiased sampling and estimation
scheme satisfying our goals (i)--(iii) that we recall below.
\begin{description}
\item[\textnormal{(i)}] Ipps. In the rest of the paper, we 
use the threshold centric definition from
\cite{DLT05} mentioned under ipps Poisson sampling
in Section \ref{sec:all-schemes}. Thus we have the sampling probabilities 
$p_i=\min\{1,w_i/\tau_k\}$ where $\tau_k$ is the unique value such 
that $\sum_{i\in[n]}\min\{1,w_i/\tau_k\}=k$ assuming $k<n$; otherwise
$\tau_k=0$ meaning that all items are sampled. The expected
number of samples is thus $\min\{k,n\}$.
A sampled item $i$ gets
the Horvitz-Thompson estimator $w_i/p_i=\max\{w_i,\tau_k\}$.
We refer to $\tau_k$ as
{\em the threshold} when $k$ and the weights are understood. 
\item[\textnormal{(ii)}] At most $k$ samples. Together with (i) this
means exactly $\min\{k,n\}$ samples.
\item[\textnormal{(iii)}] No positive covariances.
\end{description}
Recall that these properties imply all variance qualities mentioned in
the introduction.

As mentioned in the introduction, a clean
design that differentiates our
$\varoptk$ scheme from preceding schemes 
is that we can just sample from samples without relying on auxiliary 
data.
To make sense of this statement, we let all sampling scheme
operate on some adjusted weights, which initially are the original
weights. When we sample some items with adjusted weight, we
use the resulting weight estimates as new adjusted weights, treating them
exactly as if they were original weights.

\subsection{A general recurrence}\label{sec:gen-recurse}
Our main contribution is a general recurrence for generating $\varoptk$ schemes.
Let $I_1,...,I_m$ be disjoint non-empty sets of weighted items, and $k_1,...,k_m$ be
integers each at least as large as $k$. Then 
\begin{equation}\label{eq:gen-recurse}
\varoptk(\bigcup_{x\in[m]}I_x)=\varoptk(\bigcup_{x\in[m]}\varopt_{k_x}(I_x))
\end{equation}
We refer to the calls to $\varopt_{k_x}$ on the right hand side as the 
{\em inner subcalls}, the call to $\varoptk$ as the
{\em outer subcall}. The call to $\varoptk$ on the left hand
side is the {\em resulting call}. The recurrence states that
if all the subcalls are $\varoptk$ schemes (with the $k_x$ replacing $k$ 
for the inner subcalls), that is, unbiased sampling
and estimation schemes satisfying properties (i)--(iii), then the resulting call is also a $\varoptk$ scheme.
Here we assume that the random choices of different subcalls 
are {\em independent} of each other.

\subsection{Specializing to reservoir sampling}\label{sec:stream}
To make use of \req{eq:gen-recurse} in a streaming context, first as a base case,
we assume an implementation of $\varoptk(I)$ when $I$ has $k+1$ items, denoting
this procedure $\varoptkbase$. This is very simple and has been done before
in \cite{Cha82,Til96}. Specializing \req{eq:recurse} with $m=2$, $k_1=k_2=k$,
$I_1=\{1,...,n-1\}$
and $I_2=\{n\}$, we get
\begin{equation}\label{eq:recurse}
\varoptk([n])
=\varoptkbase(\varoptk([n-1])\cup\{n\}).
\end{equation}
With \req{eq:recurse} we immediately
get a $\varoptk$ reservoir sampling algorithm: the first $k$ items
fill the initial reservoir. Thereafter, whenever a new item arrives, we add
it to the current reservoir sample, which becomes of size $k+1$.
Finally we apply $\varoptkbase$ sample to the result. In the
application of $\varoptkbase$ we do not distinguish between items
from the previous reservoir and the new item.

\subsection{Relation to Chao's and Till\'e's procedures}\label{sec:Chao-Tille}
When we use \req{eq:recurse}, we generate exactly the same
distribution on samples as that of Chao's procedure
\cite{Cha82}. However, Chao does not use adjusted weights, let alone the
general recurrence. Instead,
when a new item $n$ arrives, he computes the new ipps probabilities
using the recursive formula from statics mentioned under (i) in Section 
\ref{sec:properties}. This formulation may involve details of all
the original weights even if we are only want the inclusion
probability of a given item. Comparing the new and the previous probabilities,
he finds the distribution for which item to drop. Our recurrence with
adjusted weights is simpler and more efficient because we can forget
about the past: the original weights and the inclusion probabilities
from previous rounds.

We can also use \req{eq:gen-recurse} to derive the elimination procedure
of Till\'e \cite{Til96}. To do that, we set $m=1$ and
$k_1=k+1$, yielding the recurrence 
\[\varoptk(I)=\varoptkbase(\varopt_{k+1}(I))\]
This tells us how to draw $k$ samples by eliminating the $n-k$ other
items one at the time. Like Chao, Till\'e \cite{Til96} computes
the elimination probabilities for all items in all rounds directly from the original
weights. Our general recurrence \req{eq:gen-recurse} based on adjusted weights is
more flexible, simpler, and more efficient.

\subsection{Relation to previous reservoir sampling schemes}
It is easy to see that nothing like \req{eq:recurse} works
for any of the other reservoir sampling schemes from the introduction.
E.g., if $\uniformk$ denotes
uniform sampling of $k$ items with associated estimates, then
\[\uniformk([n]\})
\neq\uniformkbase(\uniformk([n-1])\cup\{n\}).\]
With equality, this formula would say that item $n$
should be included with probability $k/(k+1)$. However,
to integrate item $n$ correctly in the uniform reservoir 
sample, we  should only include it with probability $k/n$.
The standard algorithms \cite{FMR62,Vit85} therefore
maintain the index $n$ of the last arrival. 

We have the same issue with all the other schemes: ppswr, ppswor,
priority, and Poisson ipps sampling. For each of these
schemes, we have a global description of what the reservoir should
look like for a given stream. When a new item arrives, we cannot just
treat it like the current items in the reservoir, sampling $k$ out of
the $k+1$ items. Instead we need some additional information in order
to integrate the new item in a valid reservoir sample of the new
expanded stream. In particular, priority sampling \cite{DLT07} and
the ppswor schemes of \cite{CK:podc07,CK07,bottomk:VLDB2008} use
priorities/ranks for all items in the reservoir, 
and the reservoir version of Poisson ipps sampling from \cite{DLT05, DLT07} uses the sum 
of all weights below the current threshold. 

\paragraph{Generalizing from unit weights}
%Mikkel6/25 this is supposed to cover what Edith started earlier
The standard scheme \cite{FMR62,Vit85} for sampling $k$ unit items is
variance optimal and we can see $\varoptk$ as a generalization to
weighted items which produces exactly the same sample and estimate
distribution when applied to unit weights. The standard scheme for
unit items is, of course, much simpler: we include the $n$th item with
probability $n/k$, pushing out a uniformly random old one. The
estimate of any sampled item becomes $n/k$.  With $\varoptk$, when the
$n$th item arrives, we have $k$ old adjusted weights of size
$(n-1)/k$ and a new item of weight $1$. We apply the general 
$\varoptkbase$ to get down to $k$ weights. The result of this 
more convoluted procedure ends up the
same: the new item is included with probability $1/n$, and all
adjusted weights become $n/k$.

However, $\varoptk$ is not the only natural generalization of
the standard scheme for unit weights. The ppswor schemes
from \cite{CK:podc07,CK07,bottomk:VLDB2008} also produce
the same results when applied to unit weights. However,
ppswor and $\varoptk$ diverge when the weights are not all the
same. The ppswor scheme from \cite{bottomk:VLDB2008} does have exact
total ($\VS=0$), but suboptimal $\SV$ so it is not variance optimal.

Priority sampling is also a generalization in that it produces the same
sample distribution when applied to unit weights. However,
the estimates vary a bit, and that is why it only optimizes $\SV$ modulo
one extra sample. A bigger caveat is that priority sampling does not get 
the total exact as it has $\VS=\SV$.

The \varoptk scheme is the unique generalization of the standard
reservoir sampling scheme for unit weights to general weights that
preserves variance optimality.

\subsection{Distributed and parallel settings}
Contrasting the above specialization for streams, we note that the general recurrence is useful in, say, a distributed setting,
where the sets $I_x$ are at different locations and only local samples
$\varopt_{k_x}(I_x)$ are forwarded to the take part in the global
sample. Likewise, we can use the general recurrence for fast parallel computation, cutting a huge file $I$ into segments $I_x$ that we sample
from independently.

\section{The recurrence}\label{sec:recur}
\newcommand\tauxk[2]{\tau_{{#1},{#2}}}
%\newcommand\tauxk[2]{\tau^{({#1})}_{{#2}}}
%Tille2: generalized thinks to include Tille in the end
We will now establish the recurrence \req{eq:gen-recurse} stating that
\[\varoptk(\bigcup_{x\in[m]}I_x)=\varoptk(\bigcup_{x\in[m]}\varopt_{k_x}(I_x))\]
Here $I_1,...,I_m$ are disjoint non-empty sets
of weighted items, and we have $k_x\geq k$ for each $x\in[m]$.

We want to show that if each subcall on the right hand side is a
$\varoptk$ scheme (with the $k_x$ replacing $k$ 
for the inner subcalls), that is, unbiased sampling
and estimation schemes satisfying (i)--(iii), then the resulting call is also a $\varoptk$ scheme. The hardest part is to prove (i), and we will do that
last.

Since an unbiased estimator of an unbiased estimator is an unbiased
estimator, it follows that \req{eq:gen-recurse} preserves unbiasedness.
For (ii) we just need to argue that the resulting sample is of size at most
$k$, and that follows trivially from (ii) on the outer subcall, regardless of
the inner subcalls.

Before proving (i) and (iii), we fix some notation.
Let $I=\bigcup_{x\in[m]}I_x$. We use $w_i$ to denote the original
weights. For each $x\in[m]$, set $I_x'=\varoptk(I_x)$, and use $w_i'$
for the resulting adjusted weights. Set
$I'=\bigcup_{x\in[m]}I_x'$. Finally, set $S=\varoptk(I')$ and use the
final adjusted weights as weight estimates $\hat w_i$.
Let $\tauxk{x}{k_x}$ be the threshold used in $\varoptk(I_x)$, and let
$\tau'_k$ be the threshold used by  $\varoptk(I')$.
\begin{lemma}\label{lem:non-pos} The recurrence \req{eq:gen-recurse}
preserves (iii).
\end{lemma}
\begin{proof}
With (iii) is satisfied for each inner subcall, we know that
there are no positive covariances in the adjusted weights $w'_i$ from 
$I_x'=\varoptk(I_x)$. Since these samples are independent, we get no positive
covariances in $w'_i$ of all items in $I'=\bigcup_{x\in[m]} I_x'$.
Let $(I^0,w^0)$ denote any possible concrete value of $(I',w')$.
Then
\begin{eqnarray*}
\lefteqn{\E[\hat w_i\hat w_j]}\\
&=&\sum_{(I^0,w^0)}\left(\Pr[(I',w')=(I^0,w^0)]\right.\\[-3ex]
&&\quad\quad\quad\quad\cdot \left.\E[\hat w_i\hat w_j\;|\;(I',w')=(I^0,w^0)]\right)\\
&\leq&\sum_{(I^0,w^0)}\left(\Pr[(I',w')=(I^0,w^0)]\;
w^0_iw^0_j\right)\\
&=&\E[w'_iw'_j]\ \leq\ w_iw_j.
\end{eqnarray*}
\end{proof}
To deal with (i), we need the following general consequence of (i) and (ii):
\begin{lemma}\label{lem:unique-multi} If (i) and (ii) is satisfied by a
scheme sampling $k$ out of $n$ items, then the multiset of adjusted
weight values in the sample is a unique function of the input weights.
\end{lemma}
\begin{proof} If $k\geq n$, we include all weights, and the result
is trivial, so we may assume $k<n$. We already noted that
(i) and (ii) imply that exactly $k$ items are sampled. The threshold $\tau_k$ from \req{eq:thr}
is a function of the input weights. All items with higher weights
are included as is in the sample, and the remaining sampled items all get adjusted weight $\tau_k$. 
\end{proof}
In the rest of this section, {\em we assume that each inner subcall
  satisfies (i) and (ii), and that the outer subcall satisfies
  (i)}. Based on these assumptions, we will show that (i) is satisfied
by the resulting call.  
\begin{lemma}\label{lem:tau-fixed}
The threshold $\tau'_k$ of the outer subcall is unique.
\end{lemma}
\begin{proof} We apply Lemma \ref{lem:unique-multi} to all the inner subcalls,
and conclude that the multiset of adjusted weight values in $I'$ is unique. This
multiset uniquely determines $\tau'_k$.
\end{proof}
We now consider a simple
degenerate cases.
\begin{lemma}\label{lem:triv} 
The resulting call satisfies (i)  if $|I'|=\sum_{x\in[m]}|I_x'|\leq k$.
\end{lemma}
\begin{proof}
If $|I|=\sum_{x\in[m]}|I_x|\leq k$, there is no active sampling by any
call, and then (i) is trivial.  Thus we may assume
$\sum_{x\in[m]}|I_x'|=\sum_{x\in[m]}\min\{k_x,|I_x|\} \leq
k<\sum_{x\in[m]}|I|$. This implies that $|I'_x|=k_x\geq k$ for some
$x$.  We conclude that $m=1$, $x=1$, and $k_1=k$, and this is
independent of random choices. The resulting sample is then is
identical to that of the single inner subcall on $I_1$ and we have
assumed that (i) holds for this call.
\end{proof}
In the rest of the proof, {\em we assume $|I'|>k$}. 
\begin{lemma}\label{lem:monotone} 
We have that $\tau'_k>\tauxk{x}{k_x}$ 
for each $x\in[m]$.
\end{lemma}
\begin{proof}
Since we have assumed $|I'|>k$, we have
$\tau_k'>0$. The statement is thus trivial for $x$ if $|I_x|\leq k$
implying $\tauxk{x}{k_x}=0$. However, if $|I_x|\geq k$, then from (i) and (ii) 
on the inner subcall $\varopt_{k_x}(I_x)$, we get 
that the returned $I'_x$ has exactly $k_x$ items, each 
of weight at least $\tauxk{x}{k_x}$. These items are all in $I'$.
Since $|I'|>k$,  it follows from (i) with \req{eq:thr} on
the outer subcall that $\tau_k'>\tauxk{x}{k_x}$.
\end{proof}
\begin{lemma}\label{lem:outvalues}
The resulting sample $S$ includes all $i$ with $w_i>\tau_k'$. 
Moreover, each $i\in S$ has $\hat w_i=\max\{w_i,\tau_k'\}$.
\end{lemma}
\begin{proof}
Since $\tau_k'>\tauxk{x}{k_x}$ and (i) holds for each inner subcall, 
it follows that $i\in I$ has $w'_i=w_i>\tau_k'$ if and only
if $w_i>\tau'_k$. The result now follows from (i) on the outer subcall.
\end{proof}
\begin{lemma}\label{lem:prob} 
The probability that $i\in S$ is $p_i=\min\{1,w_i/\tau_k'\}$.
\end{lemma}
\begin{proof}
From 
Lemma \ref{lem:tau-fixed} and \ref{lem:outvalues}  we get that
$\hat w_i$ equals the fixed value $\max\{w_i,\tau_k'\}$ if $i$ is
sampled. Since $\hat w_i$ is unbiased, we conclude that
$p_i=w_i/\max\{w_i,\tau_k'\}=\min\{1,w_i/\tau_k'\}$.
\end{proof}
\begin{lemma}\label{lem:good-thr} 
$\tau_k'$ is equal to the threshold 
$\tau_k$ defined directly for $I$ by (i).
\end{lemma}
\begin{proof} Since the input $I'$ to the outer subcall is more than $k$ items
and the call satisfies (i), it returns an expected number of $k$ items and
these form the final sample $S$. With $p_i$ the probability that item $i$ is 
included in $S$, we conclude that $\sum p_i=k$.
Hence by Lemma \ref{lem:prob}, we have $\sum_i \min\{1,w_i/\tau_k'\}=k$.
However, (i) defines $\tau_k$ as the unique value such that
$\sum_i \min\{1,w_i/\tau_k\}=k$, so we conclude that $\tau_k'=\tau_k$.
\end{proof}
From Lemma \ref{lem:outvalues}, \ref{lem:prob}, and \ref{lem:good-thr},
we conclude 
\begin{lemma}\label{lem:i} 
If (i) and (ii) are satisfied for each inner subcall and (i)
is satisfied by the outer subcall, then (i) is satisfied by the
resulting call.
\end{lemma}
We have now shown that the sample $S$ we generate satisfies (i), (ii),
and (iii), hence that it is a $\varoptk$ sample. Thus
\req{eq:gen-recurse} follows.

\section{Efficient implementations}\label{sec:efficient}
We will now show how to implement $\varoptkbase$. First we give a basic
implementation equivalent to the one used in \cite{Cha82,Til96}. Later
we will tune our implementation for use on a stream.

The input is a set
$I$ of $n=k+1$ items $i$ with adjusted weights $\tilde w_i$. We want a
$\varoptk$ sample of $I$.  First we compute the threshold $\tau_k$
such that $\sum_{i\in[n]}\min\{1,\tilde w_i/\tau_k\}=k$. We want to
include $i$ with probability $p_i=\min\{1,\tilde w_i/\tau_k\}$, or
equivalently, to drop $i$ with probability $q_i=1-p_i$.  Here
$\sum_{i\in I}q_i=n-k=1$. We partition the unit interval $[0,1]$ into
a segment of size $q_i$ for each $i$ with $q_i>0$. Finally, we pick a
random point $r\in [0,1]$.  This hits the interval of some $d\in I$,
and then we drop $d$, setting $S=I\setminus\{d\}$.  For each $i\in S$
with $\tilde w_i<\tau_k$, we set $\tilde w_i=\tau_k$. Finally we
return $S$ with these adjusted weights.

\begin{lemma}\label{lem:base}
\varoptkbase\ is a \varoptk\ scheme.
\end{lemma}
\begin{proof} It follows directly from the definition that we
use threshold probabilities and estimators, so (i) is satisfied.
Since we drop one, we end up with exactly $k$ so (ii) follows.
Finally, we need to argue that there are no positive
covariances. We could only have positive covariances between
items below the threshold whose inclusion probability is below $1$. 
Knowing that one such item is
included can only decrease the chance that another is included. Since
the always get the same estimate $\tau_k$ if included, we conclude
that the covariance between these items is negative. This settles (iii).
\end{proof}

\subsection{An $O(\log k)$ implementation}\label{sec:worst-case}
We will now improve $\varoptkbase$ to handle each new
item in $O(\log k)$ time. Instead of starting from scratch,
we want to maintain a reservoir with a sample $R$ of size $k$ for the 
items seen thus far. We denote by $R_j$ the a reservoir after processing
item $j$.

In the next subsection, we will show how to process each item in
$O(1)$ expected amortized time if the input stream is randomly
permuted.

Consider round $j>k$. Our first goal  is to identify the 
new threshold $\tau=\tau_{k,j}>\tau_{k,j-1}$.  
Then we  subsample $k$ out of
 the $k+1$ items in $R^{\mathrm{pre}}_j=R_{j-1}\cup\{j\}$. Let $\tilde
w_{(1)},...,\tilde w_{(k+1)}$ be the adjusted  weights of
the items in $R^{\mathrm{pre}}_j$ in increasing sorted order,
breaking ties arbitrarily.  
We first identify the largest number $t$ such that $\tilde w_{(t)}\leq
\tau$.  Here 
{\small 
\begin{eqnarray}
\tilde w_{(t)}\leq \tau & \iff & k+1-t+(\sum_{x\leq t}\tilde
w_{(x)} )/\tilde w_{(t)}\geq k \nonumber \\
 & \iff & (\sum_{x\leq t}\tilde
w_{(x)} )/\tilde w_{(t)}\geq t-1 \ .\label{eq:t}
\end{eqnarray}
}
After finding $t$ we 
find $\tau$ as the solution to
{\small
\begin{equation}\label{eq:tau}
(\sum_{x\leq t}\tilde w_{(x)})/\tau=t-1\iff \tau=
(\sum_{x\leq t}\tilde w_{(x)})/(t-1) \ .
\end{equation}
}
To find the item to leave out,
we pick a uniformly random number $r\in(0,1)$,
and find the smallest $d\leq t$ such that
{\small
\begin{equation}\label{eq:d}
\sum_{x\leq d}(1-\tilde w_{(x)}/\tau)\geq r\iff d\tau-\sum_{x\leq d}\tilde 
w_{(x)}\geq 
r\tau \ .
\end{equation}
}
Then  the $d$th smallest item in $R^{\mathrm{pre}}_j$, is the one we
drop to create the sample $S=R_j$.

The equations above suggests that we  find $t$, $\tau$, and $d$ by a
 binary search. When we consider an item during this search
 we  need to know 
the number of items of smaller adjusted weight, and their
total adjusted weight.

To perform this binary search we represent
$R_{j-1}$ divided into two sets. The
set $L$ of large items with $w_i> \tau_{k,j-1}$ and $\tilde w_i=w_i$,
and the set $T=R_{j-1}\setminus L$ of small items whose 
adjusted weight is equal to the threshold $\tau_{k,j-1}$.
We represent $L$ in sorted order by a
balanced  binary search tree.
Each node in this tree stores the number of items in
its subtree and their total  weight. 
We represent $T$ in sorted order (here in fact the order could
be arbitrary) by a balanced binary search tree,
where each node in this tree stores the number of items in
its subtree. If we multiply the number of items in a subtree
of $T$ by $\tau_{k,j-1}$ we get their total adjusted weight. 

The height of each of these two trees is $O(\log k)$
so  we can insert or delete an element,
or concatenate or split a list in $O(\log k)$ time 
\cite{CLRS01}. Furthermore, if
we follow a path down from the root of one of these trees to
a node $v$, then 
by accumulating counters from roots of subtrees 
hanging to the left of the path, and smaller nodes
on the path,
we can maintain the number of items in the tree
 smaller than the
one at $v$, and the total adjusted weight of these items.

We process item $j$ as follows.
 If  item $j$ is large, that is $w_{j}>\tau_{k,j-1}$, we
insert it into the tree representing $L$.
Then we find $t$ 
by searching the tree over
$L$
as follows. While at a node $v$ we
compute the total 
number of items smaller than the one at $v$ by
adding to the number of such items in $L$, $|T|$
or $|T| + 1$ depending upon whether  $w_{j}\le \tau_{k,j-1}$ or not.
Similarly, we compute the total adjusted weight of items 
smaller than the one at $v$ by
adding  $|T|\tau_{k,j-1}$ to the
total weight of such items $L$,
and  $w_{j}$ if  $w_{j}\le \tau_{k,j-1}$.
Then we use Equation \req{eq:t} to decide if
$t$ is the index of the item at $v$, or
we should proceed to the left or to the right child of $v$.
After computing $t$ we compute 
 $\tau$ by Equation \req{eq:tau}. Next
we identify $d$ by first considering
item  $j$ if $w_{j}<\tau_{k,j-1}$, and then
searching either the tree over $T$ or the
tree over $L$ in a way similar to the search for
computing $t$
but using Equation (\ref{eq:d}). Once finding $d$ our subsample becomes
$R_j=S=R^{\mathrm{pre}}_j\setminus\{d\}$. All this takes $O(\log k)$.

Last we  update our representation of
 the  reservoir, so that it corresponds to $R_j$ and $\tau_{k,j}$.
We insert $w_j$ into $T$ if $w_{j}\le \tau_{k,j-1}$ (otherwise
it had already been inserted into $L$). We also delete $d$ from
the list containing it.
If $w_{(t)}$ was a large weight
 we  split $L$ at $w_{(t)}$ and concatenate the prefix of $L$ to $T$.
Our balanced trees support concatenation and split in $O(\log k)$
time, so this does not affect our overall time bounds. Thus
we have proved the following theorem.
\begin{theorem}
With the above implementation, our reservoir sampling
algorithm processes each new item in $O(\log k)$ time.
\end{theorem}
In the above implementation we have assumed constant time access to real
numbers including the random $r\in (0,1)$. Real computers do not support
real reals, so in practice we would suggest using floating point
numbers with some precision $\wp\gg \log n$, accepting a fractional
error of order $1/2^\wp$.

%Mikkel-10/1
We shall later study an alternative implementation based on a standard priority
queue, but it is only more efficient in the amortized/average sense. Using
the integer/floating point priority queue from \cite{Tho07}, it
handles any $k$ consecutive items in $O(k\log\log k)$ time, hence using
only $O(\log\log k)$ time on the average per item.

\subsection{Faster on randomly permuted streams}\label{sec:amort}
We will now discuss some faster implementations in amortized and
randomized settings. First we consider the case where the input
stream is viewed as randomly permuted.

We call the processing of a new item {\em simple\/} if it is not
selected for the reservoir and if the threshold does
not increase above any of the previous large weights. We will argue
that the simple case is dominating if $n\gg k$ and the input
stream is a random permutation of the weights. Later we get a substantial speed-up by reducing 
the processing time of the simple case to a constant.

Lemma \ref{lem:prob} implies that our reservoir sampling scheme
satisfies the condition of the following simple lemma:
\begin{lemma}\label{lem:new-include} 
Consider a reservoir sampling scheme with capacity $k$ such that when any 
stream prefix $I$ has passed by, the probability that $i\in I$ is in the 
current reservoir is independent of the order of $I$. If a stream of $n$
items is randomly permuted, then the expected number of times that
the newest item is included in the reservoir is bounded by 
$k(\ln (n/k)+O(1))$.
\end{lemma}
\begin{proof}
Consider any prefix $I$ of the stream. The average probability that an 
item $i\in I$  is in the reservoir $R$ is $|R|/|I|\leq k/|I|$. 
If $I$ is randomly 
permuted, then  this is the expected  probability that the last item of $I$
is in $R$. By linearity of 
expectation, we get that the expected number of times the
newest item is included in $R$ is bounded by $k+\sum_{j=k+1}^n k/j=
k(1+H_n-H_{k+1})=k(\ln (n/k)+O(1))$.
\end{proof}
As an easy consequence, we get
\begin{lemma}\label{lem:pass-threshold} 
When we apply our reservoir sampling algorithm
 to a randomly permuted stream, the expected number of times that the threshold
passes a weight in the reservoir is bounded by $k(\ln (n/k)+O(1))$.
\end{lemma}
\begin{proof} Since the threshold is increasing, a weight in the
reservoir can only be passed once, and we know from Lemma \ref{lem:new-include} that the expected number of weights ever entering the reservoir is
bounded by $k(\ln (n/k)+O(1))$.
\end{proof}
We  now show how to perform a simple case in  constant time.
To do so, we maintain the smallest of the large weights in
the reservoir in a variable $w_\ell$.

We now start the processing of item $j$, hoping for it to be a simple
case. We assume we know the cardinality of the set $T$ of small items
in $R_{j-1}$ whose adjusted weight is the threshold 
$\tau=\tau_{k,j-1}$. 
Tentatively as in  \req{eq:tau} we compute 
\[\tau= (w_j+|T|\tau_{k,j-1})/|T|.\]
If $w_j\geq \tau$ or $\tau\geq w_\ell$, we cannot be in the simple case,
so we revert to the original
implementation. Otherwise, $\tau$ has its correct new value
$\tau_{k,j}$, and then we proceed to generate
the random number $r\in(0,1)$ from the original algorithm. If 
\[(\tau-w_j)> r\tau\textnormal,\]
we would include the new item, so we revert to the original algorithm
using this value of $r$. Otherwise, we skip item $j$. No further processing 
is required, so we are done in constant time. 
The reservoir and its division into large items in $L$ and
small items in $T$ is unchanged. However, all the 
adjusted weights in $T$ were increased implicitly when
we increased $\tau$ from $\tau_{k,j-1}$ to $\tau_{k,j}$.

\begin{theorem}\label{thm:ran-worst} A randomly permuted stream of length $n$ is
processed in $O(n+k(\log k)(\log n))$ time.
\end{theorem}
\begin{proof}
We spend only constant time in the simple cases. 
From Lemma \ref{lem:new-include} and \ref{lem:pass-threshold} we
get that the expected number of non-simple cases is at most 
$2k(\ln (n/k)+O(1))=O(k(\log (n/k))$, and we spend only $O(\log k)$ time
in these cases.
\end{proof}

\subsection{Simpler and faster amortized implementation}\label{sec:amortized}
We will present a simpler implementation of $\varoptk$ based on a
standard priority queue. This version will also handle the above
simple cases in constant time.
From a worst-case perspective, the amortized version will not be as good 
because we may spend $O(k\log\log k)$ time on processing a single
item, but on the other hand, it is guaranteed to process
any sequence of $k$ items within this time bound. Thus the
amortized/average time per item is only $O(\log\log k)$, which is
exponentially better than the previous $O(\log k)$ worst-case 
bound.

\begin{small}
\begin{algorithm}%[H]
     \caption{$\varopt_k$. The set of items in the reservoir 
     are represented as
     $R=L\cup T$. If $i\in L$, we have 
     $\hat w_i=w_i>\tau$. For all $i\in T$, we have $\hat w_i=\tau$.
     The set $L$ is in a priority queue maintaining the item of minimum
     weight. The set $T$ is in an array.}
     \label{alg:varopt}
     \SetVline \dontprintsemicolon \BlankLine
     \BlankLine
     $L\assign\emptyset$; $T\assign\emptyset$; $\tau\assign 0$\\
     \While{$|L|<k$}{include each new item $i$ in $L$ with
            its weight $w_i$ as adjusted weight $\hat w_i$.}
     \While{new item $i$ with weight $w_i$ arrives}{
       $X\assign \emptyset$ \tcc*[f]{\rm Set/array of items to be moved from 
                                         $L$ to $T$}\\
       $W\assign \tau|T|$ \tcc*[f]{\rm sum of adjusted weights in $T\cup S$}\\
       \lIf{$w_i>\tau$}{include $i$ in $L$}\\
       \Else{$X[0]\assign i$\\
             $W\assign W+w_i$}
       \While{$W\geq(|T|+|X|-1)\min_{h\in L}w(h)$}{
          $h\assign \textnormal{argmin}_{h\in L}w(h)$\\
          move $h$ from $L$ to end of $X$\\
          $W\assign W+w_h$.
       }
       $t\assign W/(|T|+|X|-1)$\\
       generate uniformly random $r\in U(0,1)$\\
       $d\assign 0$\\
       \While{$d<|X|$ and $r\geq 0$}{
          $r\assign r-(1-w_{X[d]}/t)$\\
          $d\assign d+1$
       }
       \lIf{$r<0$}{remove $X[d]$ from $X$}
       \Else{remove uniform random element from $T$}
       append $X$ to $T$.
     }
   \end{algorithm}
\end{small}

Algorithm \ref{alg:varopt} contains the pseudo-code for the
amortized algorithm.
The simple idea is to use a priority queue for the set $L$ of large items,
that is, items whose weight exceeds the current threshold $\tau$. The
priorities of the large items are just their weight. The priority
queue provides us the lightest large item $\ell$ from $L$ 
in constant time. Assuming integer or floating point representation, we can
update the priority queue $L$ in $O(\log\log k)$ time \cite{Tho07}.
The items in $T$ are maintained in an initial segment of an array with
capacity for $k$ items.

We now consider the arrival of a new item $j$ with weight $w_j$,
and let $\tau_{j-1}$ denote the current threshold. All items
in $T$ have adjusted weight $\tau_{j-1}$ while all other
weight have no adjustments to their weights. We will build
a set $X$ with items outside $T$ that we know
are smaller than the upcoming threshold $\tau_j>\tau_{j-1}$. To 
start with, if $w_j\leq \tau_{j-1}$, we set $X=\{j\}$; otherwise we set 
$X=\emptyset$ and add item $j$ to $L$.
We are going to move items from $L$ to $X$ until $L$ only contains
items bigger than the upcoming threshold $\tau_j$. For that
purpose, we will maintain the sum $W$ of adjusted weights in 
$X\cup T$. The sum over
$T$ is known as $\tau_{j-1}|T|$ to which we add $w_j$ if $X=\{j\}$.

The priority queue over $L$ provides us
with the lightest item $\ell$ in $L$. From \req{eq:t} we know that 
$\ell$ should be moved to $X$ if and only if
\begin{equation}\label{eq:move}
W\geq w_\ell(|X|+|T|-1).
\end{equation}
If \req{eq:move} is satisfied, we delete $\ell$ from $L$ and insert
it in $X$ while adding $w_\ell$ to $W$. We repeat these
moves until $L$ is empty or we get a contradiction to \req{eq:move}.

We can now compute the new threshold $\tau_j$ as
\[\tau_j=W/(|X|+|T|).\]
Our remaining task is to find an item to be deleted based a uniformly
random number $r\in(0,1)$. If the total weight $w_X$ in $X$ is such
that $|X|-w_X/\tau_j\leq r$, we delete an item from $X$ as follows. 
With $X$ represented
as an array. Incrementing $i$ starting from $1$, we
stop as soon as we get a value such that $i-w_{X[1..i]}/\tau_j\geq r$, and
then we delete $X[i-1]$ from $X$, replacing it by the last item
from $X$ in the array.

If we do not delete an item from $X$, just delete
a uniformly random item from $T$. Since $T$ fills
an initial segment of an array, we just generate
a random number $i\in [|X|]$, and set
$X[i]=X[|T|]$. Now $|T|$ is one smaller.

Having discarded an item from $X$ or $T$, we move all remaining items 
in $X$ to the array of $T$, placing them behind the current items in $T$.
All members of $T$ have the new implicit adjusted weight $\tau_j$. We
are now done processing item $j$, ready for the next item to arrive.

\begin{theorem}\label{thm:ran-amort}  The above implementation processes
items in $O(\log\log k)$ time amortized time when averaged
over any $k$ consecutive items. Simple cases are handled in constant time,
and are not part of the above amortization. As
a result, we process a randomly permuted stream of length $n$ 
in $O(n+k(\log\log k)(\log n))$ expected time.
\end{theorem}
\begin{proof}
First we argue that over $k$ items, 
the number of priority queue updates for $L$ 
is $O(k)$. Only new items are inserted in $L$ and we started
with at most $k$ items in $L$, so the total number of updates
is $O(k)$, and each of them take $O(\log\log k)$ time. The
remaining cost of processing a given item $j$ is 
a constant plus $O(|X|)$ where $X$ may include the new
item $j$ and items taken from $L$. We saw above that we
could only take $O(k)$ items from $L$ over the processing of
$k$ items.

Now consider the simple case where we get a new light item $i$ with 
$w_i<\tau$ and where $w_i+\tau|T|< \min L|T|$. In this case,
no change to $L$ is needed. We end up with $X=\{i\}$, and then 
everything is done in constant
time. This does not impact our amortization at all.
Finally, we derive the result for randomly permuted sequences
as we derived Theorem \ref{thm:ran-worst}, but exploiting
the better amortized time bound of $O(\log\log k)$ for the
non-simple cases.
\end{proof}

\section{An $\Omega(\log k/\log\log k)$ time worst-case lower bound}
\label{sec:lower-bound}

Above we have a gap between the $\varoptk$ implementation from Section
\ref{sec:worst-case} using balanced trees to process each item in
$O(\log k)$ time, and the $\varoptk$ implement from Section
\ref{sec:amortized} using priority queues to process the items in
$O(\log\log k)$ time on the average. These bounds assume that we use
floating point numbers with some precision $\wp$, accepting a
fractional error of order $1/2^\wp$. For other weight-biased reservoir
sampling schemes like priority sampling \cite{Tho07}, we know how to
process every item in $O(\log\log k)$ worst-case time. Here we will
prove that such good worst-case times are not possible for
$\varoptk$. In particular, this implies that we cannot hope to get a
good worst-case implementation via priority queues that processes each and 
every item with only a constant number of priority queue operations.

We will prove a lower bound of $\Omega(\log k/\log\log k)$ on the worst-case
time needed to process an item for $\varoptk$. The lower-bound is
in the so-called cell-probe model, which means that it only
counts the number of memory accesses. In fact,
the lower bound will hold for any scheme that
satisfies (i) to minimize $\SV$. For a stream with more then $k$ items,
the minimal estimator in the reservoir should 
be the threshold value $\tau$ such that
\begin{equation}\label{eq:worst} \sum_{i\in R}\min\{1,w_i/\tau\}=k\iff
\sum\{w_i|i\in R, w_i\leq \tau\}=\tau(k-
|\{w_i|i\in R, w_i>\tau\}|).\end{equation}
\paragraph{Dynamic prefix sum}
Our proof of the $\varoptk$ lower bound will be by reduction from 
the dynamic prefix sum
problem: let $x_0,...,x_{k-1}\in \{-1,0,1\}$ be variables, all
starting as $0$. An update of $i$ sets $x_i$ to some value, and a
prefix sum query to $i$ asks for $\sum_{h\leq i} x_i$. We consider
``set-all-ask-once'' operation sequences where we first set every
variable exactly once, and then perform an arbitrary 
prefix sum query.
\begin{lemma}[{\cite{AHR98,FS89}}]\label{lem:restr-prefix}
No matter how we represent, update, and query information, there
will be set-all-ask-once operation sequences where some
operation takes $\Omega(\log k/\log\log k)$ time.
\end{lemma}
The above lemma can be proved with the chronograph method of Fredman
and Saks \cite{FS89}. However, \cite{FS89} is focused on amortized
bounds and they allow a mix of updates and queries.  Instead of
rewriting their proof to get a worst-case bound with a single query at
the end, we prove Lemma \ref{lem:restr-prefix} by reduction from a
marked ancestor result of Alstrup et al. \cite{AHR98}.  The reduction
was communicated to us by Patrascu \cite{Pat09}.
\begin{proof}[ of Lemma \ref{lem:restr-prefix}]
For every $k$, Alstrup et al. \cite{AHR98} shows that there is a fixed
rooted tree with $k$ nodes and a fixed traversal sequence of the
nodes, so that if we first assign arbitrary 0/1 values $b_v$ to the
nodes $v$, and then query sum over some leaf-root-path, then no matter
how we represent the information, there be a sequence of assignments
ended by a query such that one of the operations take $\Omega(\log
k/\log\log k)$ time.

To see that this imply Lemma \ref{lem:restr-prefix}, consider
a sequence $x_0,...x_{2n-1}$ corresponding to an Euler tour of
the tree. That is, first we visit the root, then we visit
the subtrees of the children one by one, and then we end at
the root. Thus visiting each node twice, first down and later
up. For node $v$ let $v^\downarrow$ be the Euler tour index
of the first visit, and $v^\uparrow$ be the index of the last
visit.

Now, when we set $b_v$ in the marked ancestor problem, we
set $x_{v^\downarrow}=b_v$ and $x_{v^\uparrow}=-b_v$. Then for
any leaf $v$, the sum of the bits over the leaf-root-path is exactly
the prefix sum $\sum_{h\leq v^\downarrow}x_h$. Hence Lemma \ref{lem:restr-prefix}
follows from the marked ancestor lower-bound of Alstrup et al. \cite{AHR98}.
\end{proof}

\paragraph{From prefix-sum to $\varoptk$}
We will now show how $\varoptk$ can be used to solve the prefix-sum problem.
The construction is quite simple. Our starting point is the
set-all-ask-once prefix-sum problem with $k$ values 
$x_0,...,x_{k-1}$. When we want
to set $x_i$, we add an item $i$ with weight $w_i=4k^3+4ki+x_i$. Note
that these items can arrive in any order. To query prefix $j$,
we essentially just add a final item $k$ with $w_{k+1}=2k(i+1)/2$, 
and ask for the threshold $\tau$ which is the adjusted weight
of item $0$ or item $k$, whichever is not dropped from the sample.

The basic point is that our weights are chosen such that the threshold $\tau$ 
defined in \req{eq:worst} must be in $(w_i,w_{i+1})$, implying
that
\[\tau=(w_k+\sum_{h\leq i} w_h)/i\ .\]
Since we are using floating point numbers, there may be some
errors, but with a precision $\wp\geq 4\log k$ bits, we
get $(w_k+\sum_{h\leq i} w_h)$ if we multiply by $i$ 
and round to the nearest integer. Finally, we
take the result modulo $3k$ to get the desired prefix sum 
$\sum_{h\leq i} x_i$. 

In our cell-probe model, the derivation of the prefix $\sum_{h\leq i} x_i$. from
$\tau$ is free since it does not involve memory access. From Lemma
\ref{lem:restr-prefix} we know that one of the prefix updates or the
last query takes $\Omega(\log k/\log\log k)$ memory
accesses. Consequently we must use this many memory accesses on our instance
of $\varoptk$, either in the processing of one of the items, or
in the end when we ask for the threshold $\tau$ which is the
adjusted weights of item $0$ or
item $k$. Hence we conclude
\begin{theorem}\label{thm:lower-bound}  No matter how we implement
$\varoptk$, there are sequences of items such that
the processing of some item takes $\Omega((\log k)/(\log\log k))$ time.
\end{theorem}

\section{Some experimental results on Netflix  data}  \label{exp:sec}
 We illustrate both the usage and 
the estimate quality attained by
$\varoptk$ through an example on a real-life data set.
The Netflix Prize~\cite{netflix} data set
consists of reviews of 17,770 distinct movie titles by $5\times 10^5$
reviewers.  The weight we assigned to each movie title is the 
corresponding number of reviews. 
%Mikkel-10/1
We experimentally compare $\varopt$ to state of the art reservoir sampling
methods.
All methods produce a fixed-size sample of
$k=1000$ titles along with an assignment of adjusted weights to 
included titles.  
These summaries (titles and adjusted weights)
support unbiased estimates on the 
weight of subpopulations of titles specified by arbitrary selection
predicate. Example selection predicates are
``PG-13'' titles, ``single-word'' titles, or ``titles released
in the 1920's''.    An estimate of the total number of reviews
of a subpopulation
is obtained by applying the selection predicate to all titles included in the
sample and summing the adjusted weights over titles for which the
predicate holds.

We partitioned the titles into subpopulations and 
computed the sum of the square errors of the estimator over the partition.
 We used natural set of partitions based
on ranges of release-years of the titles (range sizes of 1,2,5,10 years).
Specifically, for partition with range size $r$, a title with
release year $y$ was mapped into a subset containing
all titles whose release year is $y\; {\rm mod}\; r$.  We also used the
value $r=0$ for single-titles
(the finest partition).

  The methods
compared are priority sampling (\pri)~\cite{DLT07}, ppswor (probability
proportional to size
sampling with replacement) with the rank-conditioning estimator (\ws\
\rc)~\cite{CK07,bottomk:VLDB2008}, ppswor 
with the subset-conditioning estimator (\ws\
\ssc)~\cite{CK07,bottomk:VLDB2008}, and \varopt.  
We note that \ws\ \ssc\  dominates (has smaller variance on all distributions 
and subpopulations)
\ws\ \rc, which in turn, dominates
the classic ppswr Horvitz-Thomson estimator~\cite{CK07,bottomk:VLDB2008}.
Results are shown in
Figure~\ref{netflix:fig}.

The \pri\ and \ws\ \rc\ estimators have zero covariances, 
%Mikkel: In the plot is it clearly not a horizontal line
%so I added a footnote 
and therefore, as Figure \ref{netflix:fig} shows\footnote{The slight
increase disappears as we average over more and more runs.}, 
the sum of square errors is
invariant to the partition (the sum of variances is 
equal to $\SV$).

The \ws\ \ssc\ and \ws\ \rc\ estimators have the same $\SV$
and \pri~\cite{Sze06} has nearly the same $\SV$ as the optimal $\varopt$.
Therefore, as the figure shows, on single-titles ($r=0$), \ws\ \rc\ performs
the same as \ws\ \ssc\ and \pri\ performs (essentially) as well as $\varopt$.
Since $\varopt$ has optimal (minimal) $\SV$, it outperforms all other
algorithms.

% Similarly, $\varopt$ essentially dominates \pri\footnote{Because \pri\ does not strictly minimize $\SV$, 
%completely precise: it is possible that there is some distribution and
%item for which
%\pri\ has lower variance than \varopt, of course, 
%$\varopt$ has lower sum of variances.}
% Edith, I'd like to be precise, on the average \varopt dominates anything not just pri.
%and removed the footnote. Edith, feel free to put it back in.

 We next turn to larger subpopulations. Figure \ref{netflix:fig} illustrates
that for
$\varopt$ and the \ws\ \ssc, the sum of square errors 
{\em decreases} with subpopulation size and therefore they have
significant benefit over \pri\ and \ws\ \rc.
We can see that $\varopt$, that has optimal average variance for any
subpopulation size outperforms \ws\ \ssc.

% Edith, technically, \pri\ is strictly worse on SV than \varopt, so it is not precise to say that it is ``matched''
 To conclude,
$\varopt$ is the winner, 
being strictly better
than both \pri\ and \ws\ \ssc. 
In Appendix~\ref{sec:ppswor-bad} we provide theoretical
examples where \varoptk\ has a variance that is a factor $\Omega(\log k)$ smaller than that of {\em any\/} ppswor
scheme, \ws\ \ssc\ included, so the performance gains of \varoptk\ can be much larger
than on this particular real-life data set.

\begin{figure}[htbp]
{\small \centerline{
\begin{tabular}{c}
\epsfig{figure=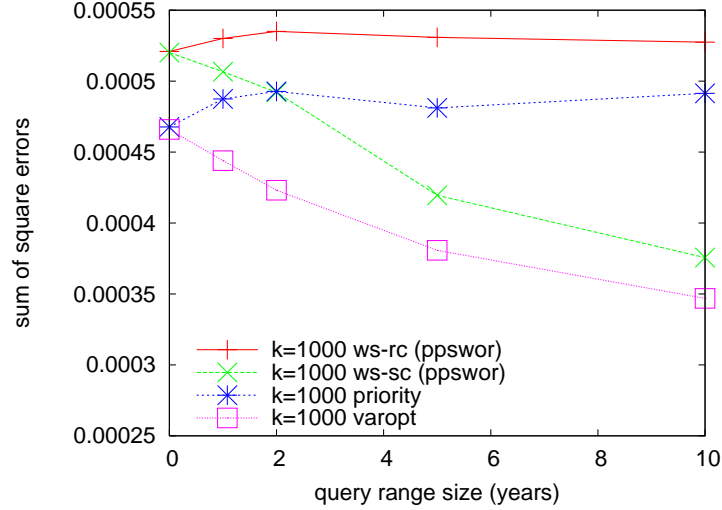,width=0.6\textwidth} 
\end{tabular}}
}
% \caption{Left: Sum of the square errors of the estimate over each partition, averaged over 500 repetitions of the summarization method.  Right: The ratio of the average sum of squared errors to that obtained using priority sampling.}
\caption{Sum of the square errors of the estimates over each partition, averaged over 500 repetitions of the respective summarization method.}
\label{netflix:fig}
\end{figure}

\section{Chernoff bounds}\label{sec:Chernoff}
In this section we will show that the $\varopt_k$ schemes from
Section \ref{sec:recur} provide estimates whose
deviations can be bounded with the Chernoff bounds usually associated
with independent Poisson samples.

Recall that we defined a $\varopt_k$ as a sampling scheme with unbiased 
estimation such that given items $i\in[n]$ with weights $w_i$:
\begin{description}
\item[\textnormal{(i)}] Ipps. We have the sampling probabilities 
$p_i=\min\{1,w_i/\tau_k\}$ where the threshold $\tau_k$ is the unique value 
such 
that $\sum_{i\in[n]}\min\{1,w_i/\tau_k\}=k$ assuming $k<n$; otherwise
$\tau_k=0$ meaning that all items are sampled. 
A sampled item $i$ gets
the Horvitz-Thompson estimator $\hat w_i=w_i/p_i=\max\{w_i,\tau_k\}$.
\item[\textnormal{(ii)}] At most $k$ samples---this hard capacity constraint
prevents independent Poisson samples.
\item[\textnormal{(iii)}] No positive covariances.
\end{description}
In Section \ref{sec:recur} we noted that when $n=k+1$, there is only a
unique $\varoptk$ scheme; namely $\varoptkbase$ which drops one
item which is item $i$ with probability $q_i=1-p_i$ with $p_i$ the
ipps from (i). We also proved that the $\varoptk$ conditions
(i)--(iii) are preserved by recurrence \req{eq:gen-recurse} stating
that
\[\varoptk(\bigcup_{x\in[m]}I_x)=\varoptk(\bigcup_{x\in[m]}\varopt_{k_x}(I_x))\ ,\]
where $I_1,...,I_m$ are disjoint non-empty sets
of weighted items and $k_x\geq k$ for each $x\in[m]$. 

In this section, we will show that any scheme generated as above satisfies
property (iii$+$) below which can be seen as a higher-order version (iii).

\paragraph{(iii$+$)} {\em High-order inclusion and exclusion probabilities are bounded by the respective product of first-order probabilities}. More precisely
for any $J\subseteq [n]$, 
\begin{eqnarray*}
\mbox{(I):}\quad\quad p[J] & \leq & \prod_{i\in J} p_i \\
\mbox{(E):}\quad\quad q[J] & \leq & \prod_{i\in J} q_i 
\end{eqnarray*}
where $p_i$ is the probability that item $i$ is included
in the sample $S$ and $p[J]$ is the probability that all $i\in J$
are included in the sample. Symmetrically
$q_i$ is the probability that item $i$ is excluded 
from the sample $S$ and $q[J]$ is the probability that all $i\in J$
are excluded from the sample. 
We will use $X_i$ as the indicator variable for $i$ being in the sample,
so $\Pr[X_i=1]=p_i$ and $\Pr[X_i=0]=q_i$.

It is standard that a special case of
Property (iii$+$)(I) implies (iii): For any $i,j$,
$p_{i,j}\leq p_i p_j$ 
combined with Horvitz-Thompson estimators implies nonnegative covariance
between $\hat w_i$ and $\hat w_j$.

The significance of (iii$+$) was argued by Panconesi and
Srinivasan~\cite{PancSri:sicomp97} who used it to prove Chernoff
bounds that we usually associate with independent Poisson
sampling. They did not consider input weights, but took the inclusion
probabilities $p_i$ as input. For us this corresponds to the special
case where the weights sum to $k$ for then we get $p_i=w_i$. Given the
inclusion probabilities $p_i$, Srinivasan \cite{Sri01} presented an
off-line sampling scheme realizing (i)-(iii+). Strengthening the
results from~\cite{PancSri:sicomp97,Sri01} slightly, we prove
\begin{theorem} \label{thm:sri}
Let $I\subseteq [n]$ and $m=|I|$. For $i\in I$, let $X_i$ be a random 0/1 variable which is $1$ with
probability $p_i$ and $0$ otherwise. The variables may not be
independent. Let
$X_I=\sum_{i\in I} X_i$, and $\mu = E[X]=\sum_{i\in I} p_i$. Finally
let $0 < a < m$.
\begin{itemize}
\item[(I)] If (iii$+$)(I) is satisfied and $a\geq\mu$, then 
\begin{equation} \label{chernoff:upper}
\Pr[X_I \ge a] \leq
\left(\frac{m-\mu}{m-a}\right)^{m-a}\left(\frac{\mu}{a}\right)^a
\quad 
\left[\leq
e^{a-\mu}\left(\frac{\mu}{a}\right)^a
\right]\ .
\end{equation}
\item[(E)] If (iii$+$)(E) is satisfied and $a\leq\mu$, then 
\begin{equation} \label{chernoff:lower}
\Pr[X_I \le a] \leq
\left(\frac{m-\mu}{m-a}\right)^{m-a}\left(\frac{\mu}{a}\right)^a
\quad 
\left[\leq
e^{a-\mu}\left(\frac{\mu}{a}\right)^a
\right]\ .
\end{equation}
\end{itemize}
\end{theorem}
\subsection{Relevance to $\varoptk$ estimates}\label{sec:relevance}
The above Chernoff bounds only address the number of sampled
items while we are interested in estimates of the weight
of some subset $H\subseteq[n]$. We split $H$ in a set of heavy
items $L=\{i\in H|w_i\geq \tau_k\}$ and a set of light items $I=H\setminus L=\{i\in H|w_i<\tau_k\}$. Then
our estimate can be written as 
\[\hat w_H=\sum_{i\in L} w_i
+\tau_k|S\cap I|= \sum_{i\in L} w_i +\tau_k\sum_{i\in I}X_i=
w_L+\tau_k X_I\ .\] Here $X_I$ is the only variable part of the
estimate and Theorem \ref{thm:sri} bounds the probability of deviations
in $X_I$ from its mean $\mu$. Using these bounds we can easily derive
confidence bounds like those in~\cite{Tho06} for
threshold sampling.

\subsection{Satisfying (iii$+$)}
In this subsection, we will show that (iii$+$) is
satisfied by all the $\varoptk$ schemes generated with
our recurrence \req{eq:gen-recurse}. As special cases this includes 
our $\varoptk$ scheme for streams
and the schemes of Chao \cite{Cha82}
and Till\'e \cite{Til96} schemes. We note
that \cite{Cha82,Til96} stated only (iii) but
(iii$+$)(I) directly follows from the expressions they provide for 
inclusion probabilities.  Condition (iii+) for
second order exclusion follows from second
order inclusion.  In \cite{Cha82,Til96} there is no mentioning and 
no derivation towards establishing (iii$+$)(E) which is much
harder to establish than (iii$+$)(I).

\begin{lemma}  \label{iiiplusforbase}
$\varopt_{k,k+1}$ satisfies (iii$+$).
\end{lemma}
\begin{proof}
Here $\varopt_{k,k+1}$ is the unique $\varoptk$ scheme when $n=k+1$.
It removes one item $i\in [k+1]$ according to
$q_i=1-p_i$ where $p_i$ are the ipps probabilities from (i).
Here $\sum_i p_i=k$ so $\sum_i q_i=1$.  
Consider $J\subseteq [k+1]$.  If $|J|=1$, (iii$+$) trivially holds.
If $|J|>1$, $q[J]=0$ and hence $q[J]\leq \prod_{i\in J} q_i$, establishing 
(iii$+$)(E).  Also 
$p[J]=1-\sum_{i\in J} q_j \leq \prod_{j\in J} (1-q_j)=\prod_{j\in J} p_j$, establishing 
(iii$+$)(I).
\end{proof}

The rest of this subsection is devoted to prove that
(iii$+$) is preserved by \req{eq:gen-recurse}.
\begin{theorem} \label{iiiplusforrecurrence}
$\varopt$ defined by (i)-(iii$+$) satisfies recurrence
\req{eq:gen-recurse}:
\[\varoptk(\bigcup_{x\in[m]}I_x)=\varoptk(\bigcup_{x\in[m]}\varopt_{k_x}(I_x))\ .\]
where $I_1,...,I_m$ are disjoint and $k_1,...,k_m\geq k$.
\end{theorem}
We want to show that if each of the subcalls on the right hand side
satisfy (i)-(iii$+$), then so does the resulting call.
In Section \ref{sec:gen-recurse} we proved this for (i)-(iii) is preserved.
It remains to prove that the resulting call satisfies (iii$+$).

We think of the above sampling as divided in two stages.
We start with $I=\bigcup_{x\in[m]}I_x$. {\em Stage (0)\/} is the
combined inner sampling, taking us from $I$ to
$I^{(1)}=\bigcup_{x\in[m]}\varopt_{k_x}(I_x)$. {\em Stage (1)\/} is the outer
subcall taking us from $I^{(1)}$ to the final sample $S=\varoptk(I^{(1)})$.

We introduce some notation.
For every $i\in I$, we denote by $p^{(0)}[i]$ 
the probability that
$i$ is included in $I^{(1)}$.
Then $q^{(0)}[i]=1-p^{(0)}[i]$ is the corresponding exclusion probability.
Observe that $p^{(0)}[i]$ and $q^{(0)}[i]$ are the respective inclusion and
exclusion probabilities of $i$ in $\varopt_{k_x}(I_x)$ for the unique
$x$ such that $i\in I_x$.  For $J\subseteq I$, 
we use the notation $p^{(0)}[J]=\Pr[J\subseteq I^{(1)}]$ and
$q^{(0)}[J]=\Pr[J\cap I^{(1)}=\emptyset]$ 
for the inclusion and exclusion probabilities of $J$ in $I^{(1)}$.
Denote by $p^{(1)}[J|I^{(1)}]$ and $q^{(1)}[J|I^{(1)}]$  
the inclusion and exclusion probabilities of $J$ by
$\varoptk(I^{(1)})$.
We denote by $p[J]$ the probability that all items in $J$ are selected 
for the final sample $S$, and by 
$q[J]$ the probability that no
item of $J$ is selected for $S$. 

Our goal is to show that (iii$+$) is satisfied for the final
sample $S$. As a first easy step exercising our nation, 
we show that (iii$+$) satisfied 
for the sample $I^{(1)}$ resulting from stage (0).
\begin{lemma} \label{lem:union}
If the samples of each independent subcall $\varopt_{k_x}(I_x)$
satisfies (iii$+$)(I) (respectively, (iii$+$)(E)),
then so does their union $I^{(1)}=\bigcup_{x\in[m]}\varopt_{k_x}(I_x)$ 
as a sample of $I$.
\end{lemma}
\begin{proof}
We consider the case of inclusion. Consider $J\subseteq I$.  
Since $\varopt_{k_x}(I_x)$ are independent, we get
$p^{(0)}[J]=\prod_x p^{(0)}[J_x]$ where
$J_x = J\cap I_x$.
We assumed (iii$+$)(I) for each $\varopt_{k_x}(I_x)$
so $p^{(0)}[J_x]\leq \prod_{i\in J_x}p^{(0)}[i]$.  Substituting we obtain
$p^{(0)}[J]=\prod_{i\in J} p^{(0)}[i]$.   The proof for exclusion
probabilities is symmetric.
\end{proof}
The most crucial property we will need from (i) and (ii) is 
a certain kind of consistency.
Assume some item $i\in I$ survives stage (0) and ends
in $I^{(1)}$. We say that overall sampling is {\em consistent\/} if the probability $p^{(1)}[i|I^{(1)}]$ that
$i$ survives stage (1) is independent of which other
items are present in $I^{(1)}$. In other words, given any two
possible values $I^{(1)}_1$ and $I^{(1)}_2$ of $I^{(1)}$,
we have $p^{(1)}[i|I^{(1)}_1]=p^{(1)}[i|I^{(1)}_2]$, and we
let $p^{(1)}[i]$ denote this unique value. Under consistency,
we also define $q^{(1)}[i]=1-q^{(1)}[i]$, and get some
very simple formulas for the overall inclusion and
exclusion probabilities; namely that $p[i]=p^{(1)}[i]\, p^{(0)}[i]$ and
$q[i]=q^{(0)}[i] + p^{(0)}[i]\,q^{(1)}[i]$.

Note that even with consistency, when $|J|>1$,
$q^{(1)}[J|I^{(1)}]$ and 
$p^{(1)}[J|I^{(1)}]$ may depend on $I^{(1)}$.
\begin{lemma} \label{consist_lemma}
Consider \req{eq:gen-recurse} where the inner subcalls
satisfy properties (i) and (ii) 
and the outer subcall satisfies property (i).
Then we have consistent probabilities $p^{(1)}[i]$ and 
$q^{(1)}[i]$ as defined above.
\end{lemma}
\begin{proof}
Our assumptions are the same as those for Lemma \ref{lem:tau-fixed}, so
we know that the threshold $\tau'$ of the outer subcall is a unique
function of the weights in $I$. 
Consider any $i\in I$, and let $x$ be
unique index such that $i\in I_x$.  We assume that
$i\in\varopt_{k_x}(I_x)$, and by (i), we
get an adjusted weight of $w^{(1)}_i=\max\{w_i,\tau_{k_x}\}$.
This value is a function of the weights in
$I_x$, hence independent of which other items are included
in $I^{(1)}$. Be by (i) on the outer
subcall, we have that the probability that $i\in I^{(1)}$ survives
the final sampling is $\min\{1,w^{(1)}_i/\tau'\}=
\min\{1,\max\{w_i,\tau_{k_x}\}/\tau'\}$, hence a direct function
of the original weights in $I$.
\end{proof}
By Lemma \ref{lem:union} and \ref{consist_lemma}, the following implies 
Theorem~\ref{iiiplusforrecurrence}.
\begin{proposition} \label{twostage}
Consider consistent two stage sampling. 
If both stages satisfy (iii$+$)(I) (resp., (iii$+$)(E)), then so 
does the composition.
\end{proposition}
As we shall see below, the inclusion part of Proposition~\ref{twostage} is
much easier than the exclusion part.
  
For  $J'\subseteq J\subseteq I$, we denote by
$\overline{p}^{(0)}[J',J]$ the probability that
items $J'$ are included and items $J\setminus J'$ are excluded
by stage (0) sampling.
In particular, $\overline{p}^{(0)}[I^{(1)},I]$ is the probability
that the outcome of stage (0) is $I^{(1)}$.
Note that we always have $\overline{p}^{(0)}[J',J]\leq p^{(0)}[J']$.
We  now  establish the easy inclusion part (I) of Proposition~\ref{twostage}. 
\begin{proof}[ of Proposition \ref{twostage}(I)]
\begin{eqnarray*}
p[J] & = & \sum_{I^{(1)}\mid J\subseteq I^{(1)}} \overline{p}^{(0)}[I^{(1)},I]p^{(1)}[J | I^{(1)}] \\
     & \le & \sum_{I^{(1)}\mid J\subseteq I^{(1)}} \overline{p}^{(0)}[I^{(1)},I]\prod_{j\in J}p^{(1)}[j] \\
     & = &  \prod_{j\in J}p^{(1)}[j]  p^{(0)}[J]\\
     & \leq  &  \prod_{j\in J}p^{(0)}[j] \prod_{j\in J}p^{(1)}[j] \\
     & = &  \prod_{j\in J}p^{(0)}[j] p^{(1)}[j] =\prod_{j\in J}p[j]
\end{eqnarray*}
\end{proof}

Next we establish the much more tricky exclusion part of 
Proposition~\ref{twostage}. First we need
\begin{lemma} \label{incexc:lemma}
\begin{equation}\label{eq:incexc}
\overline{p}^{(0)}[J',J] = \sum_{H\subseteq J'} (-1)^{|H|}q^{(0)}[H\cup (J\setminus J')]\end{equation}
\end{lemma}
\begin{proof}
Let $B$ be the event that $J\setminus J'$ are excluded from the sample.
For $j\in J'$, let $A_j$ be the event 
that $\{j\}\cup J\setminus J'$ are excluded from the sample..

 From definitions,
\begin{equation} \label{incexc1}
\overline{p}^{(0)}[J',J]  =  \Pr[B]-\Pr[\bigcup_{j\in J'} A_j]\ .
\end{equation}
 Applying the general inclusion exclusion principle, we
obtain
\begin{eqnarray}
\Pr[\bigcup_{j\in J'} A_j] & = & \sum_{H \,|\, \emptyset\neq H\subseteq J'} (-1)^{|H|+1} \Pr[\bigcap_{j\in H} A_j]\nonumber \\
 & = & \sum_{H\,|\,\emptyset\neq H\subseteq J'} (-1)^{|H|+1} q^{(0)}[H\cup (J\setminus J')]\label{incexcsub2} 
\end{eqnarray}
Now \req{eq:incexc} follows using $\Pr[B]=q^{(0)}[J\setminus J']$ and \req{incexcsub2}
in \req{incexc1}.
\end{proof}
We are now ready to establish the exclusion part (E) 
of Proposition \ref{twostage} stating that with consistent
two stage sampling, if both stages satisfy (iii$+$)(E),
then so does the composition.
\begin{proof}[ of Proposition \ref{twostage}(E)]

We need to show that $q[J] \leq \prod_{j\in J}q[j]$
\begin{eqnarray} 
q[J] & = &\sum_{I^{(1)}, J'=J\cap I^{(1)}}\overline{p}^{(0)}[I^{(1)},I]\,q^{(1)}[J'\,|\,I^{(1)}]\nonumber \\
& = &\sum_{J'\subseteq J} \left(\sum_{I^{(1)}\,|\, I^{(1)} \cap J=J'} \overline{p}^{(0)}[I^{(1)},I]q^{(1)}[J'|I^{(1)}]\right)\nonumber \\
& \leq &\sum_{J'\subseteq J} \left(\sum_{I^{(1)}\,|\, I^{(1)} \cap J=J'} \overline{p}^{(0)}[I^{(1)},I]\prod_{j\in J'} q^{(1)}[j|I^{(1)}]\right)\label{appl:E1} \\
& = &\sum_{J'\subseteq J} \left(\sum_{I^{(1)}\,|\, I^{(1)} \cap J=J'} \overline{p}^{(0)}[I^{(1)},I]\prod_{j\in J'} q^{(1)}[j]\right)\label{appl:consist1} \\
& = &\sum_{J'\subseteq J}\left(\sum_{I^{(1)}\,|\, I^{(1)} \cap J=J'} 
\overline{p}^{(0)}[I^{(1)},I]\right)\prod_{j\in J'} q^{(1)}[j]\nonumber \\
& = & \sum_{J'\subseteq J} \overline{p}^{(0)}[J',J] \prod_{j\in J'} q^{(1)}[j] \label{eq_exc} 
\end{eqnarray}
Above, for \req{appl:E1}, we applied (iii$+$)(E) on stage (1), and
for \req{appl:consist1}, we applied consistency.
We now apply Lemma~\ref{incexc:lemma} to $\overline{p}^{(0)}[J',J]$ 
in (\ref{eq_exc}), and get
\begin{eqnarray}
q[J] & \leq & \sum_{J'\subseteq J} \left(\sum_{H\subseteq J'} (-1)^{|H|}\,q^{(0)}[H\cup (J\setminus J')]\right) \prod_{j\in J'} q^{(1)}[j]\nonumber\\
 & = & \sum_{(L,H,J') | H\subseteq L\subseteq J,\; J'=H\cup (J\setminus L)} (-1)^{|H|}\,q^{(0)}[L] \prod_{j\in J'} q^{(1)}[j]\nonumber\\
 & = & \sum_{L\subseteq J} q^{(0)}[L] \prod_{j\in J\setminus L} q^{(1)}[j] 
\left(\sum_{H\subseteq L}  \prod_{h\in H} \left(-q^{(1)}[j]\right)\right)\nonumber
\end{eqnarray}
Note the convention that the empty product 
$\prod_{h\in \emptyset} \left(-q^{(1)}[h]\right)\equiv 1$. Observe that
$\sum_{H\subseteq L} \prod_{h\in H} \left(-q^{(1)}[h]\right)\ =
\ \prod_{\ell\in L}\left(1-q^{(1)}[\ell]\right)$ is non-negative.  Moreover,
from (iii$+$)(E) on stage (0), we have $q^{(0)}[L] \leq
\prod_{\ell\in L}q^{(0)}[\ell]$. Hence, we get
\begin{eqnarray}
q[J] & \leq & 
\sum_{L\subseteq J} \prod_{\ell\in L}q^{(0)}[\ell] \prod_{j\in J\setminus L} q^{(1)}[j] 
\left(\sum_{H\subseteq L}  \prod_{h\in H} \left(-q^{(1)}[j]\right)\right)\nonumber\\
 & = & \sum_{(L,H) | H\subseteq L\subseteq J}\ \prod_{\ell\in L\setminus H}q^{(0)}[\ell] 
\prod_{j\in J\setminus L} q^{(1)}[j]\prod_{h\in H}\left(- q^{(0)}[h]q^{(1)}[h]\right)\label{lhs_exc}.
\end{eqnarray}
To prove $q[J]\leq \prod_{j\in J}q[j]$, we re-express $\prod_{j\in J}q[j]$ to show that it is
equal to \req{lhs_exc}.
\begin{eqnarray}
\prod_{j\in J}q[j]& = & 
\prod_{j\in J} \left(q^{(0)}[j]+(1-q^{(0)}[j])q^{(1)}[j]\right)\nonumber\\
 &=& \prod_{j\in J} (q^{(0)}[j]+q^{(1)}[j] -q^{(0)}[j]q^{(1)}[j])\label{rhs_exc}
\end{eqnarray}
Now \req{rhs_exc} equals \req{lhs_exc} because $(L\setminus H,J\setminus L, H)$
ranges over all 3-partitions of $J$. 
\end{proof}
We have now proved both parts of Proposition \ref{twostage} which
together with Lemma \ref{lem:union} and \ref{consist_lemma} implies
Theorem~\ref{iiiplusforrecurrence}. Thus we conclude that
any $\varoptk$ scheme generated from $\varoptkbase$ and recurrence
\req{eq:gen-recurse} satisfies (i)--(iii$+$).

\subsection{From (iii$+$) to Chernoff bounds}
We now want to prove the Chernoff bounds of Theorem \ref{thm:sri}. We give
a self-contained proof but many of the calculations are 
borrowed from~\cite{PancSri:sicomp97,Sri01}. The
basic setting is as follows. Let $I\subseteq [n]$ and $m=|I|$. For $i\in I$, let $X_i$ be a random 0/1 variable which is $1$ with
probability $p_i$ and $0$ otherwise. The variables may not be
independent. Let
$X_I=\sum_{i\in I} X_i$, and $\mu = E[X]=\sum_{i\in I} p_i$. Finally
let $0 < a < m$. Now Theorem  \ref{thm:sri} falls in two statements:{\em
\begin{itemize}
\item[(I)] If (iii$+$)(I) is satisfied and $a\geq\mu$, then 
\[\Pr[X_I \ge a] \leq
\left(\frac{m-\mu}{m-a}\right)^{m-a}\left(\frac{\mu}{a}\right)^a
\]
\item[(E)] If (iii$+$)(E) is satisfied and $a\leq\mu$, then 
\[
\Pr[X_I \le a] \leq
\left(\frac{m-\mu}{m-a}\right)^{m-a}\left(\frac{\mu}{a}\right)^a\]
\end{itemize}}
We will now show that it suffices to prove Theorem  \ref{thm:sri} (I).
\begin{lemma} \label{thm:sri-}
Theorem  \ref{thm:sri} 
(I) implies Theorem  \ref{thm:sri} (E).
\end{lemma}
\begin{proof}
Define random 0/1 variables $Y_i = 1 - X_i$, $i \in [n]$, let $Y =
\sum_{i\in [n]} Y_i$, $\gamma = \E[Y]$, and $b=m-a$. Note that $Y = m - X$ and 
$\gamma = m - \mu$. We have that
$$\Pr[X \le a] =  \Pr[m-X \ge m-a] = \Pr[Y \ge b] \ . $$
Now if $X_i$ satisfies (iii$+$)(E) then $Y_i$
satisfies (iii$+$)(I) so we can apply Theorem \ref{thm:sri} (I) to $Y$ and $b$
and get 
\[
\Pr[Y \ge b]
\leq \left(\frac{\gamma}{b}\right)^{b}\left(\frac{m-\gamma}{m-b}\right)^{m-b}\\
= 
\left(\frac{m-\mu}{m-a}\right)^{n-a}\left(\frac{\mu}{a}\right)^a\ ,
\]
Hence Theorem \ref{thm:sri} (E) follows.
\end{proof}
We will now prove the Chernoff bound of Theorem \ref{thm:sri} (E). 
The traditional proofs of such bounds for $X=\sum X_i$ assume that the 
$X_i$s are independent and uses the equality
$\E[e^{tX}] = \prod_i \E[e^{tX_i}]$. Our $X_i$ are not independent, 
but we have something as good, essentially proved in \cite{PancSri:sicomp97}.
\begin{lemma}\label{lem:etS}
Let $X_1,...,X_m$ be random 0/1 variables satisfying
(iii$+$)(I), that is, for any $J\subseteq[n]$,
\[\Pr[\prod_{j\in J}X_j=1]\leq \prod_{j\in J}\Pr[X_j=1].\]
Let $I\subseteq[n]$ and $X=\sum_{i\in I} X_i$. Then for any $t\geq 0$,
\[\E[e^{tX}] \le \prod_{i\in I} \E[e^{tX_i}].\]
\end{lemma}
\begin{proof}
For simplicity of notation, we assume $I=[m]=\{1,...,m\}$.
Let $\hat{X_1},....,\hat{X_m}$ be independent random 0/1
variables with the same marginal distributions as the $X_i$, that is,
$\Pr[X_i=1]=\Pr[\hat{X_i}=1]$. 
Let $\hat X=\sum_{i\in [m]}\hat{X_i}$. 
We will prove the lemma by proving
\begin{eqnarray}
\E[e^{tX}]& \le &\E[e^{t\hat X}]\label{eq:indep}\\
&=&\prod_{i\in [m]} \E[e^{t\hat{X_i}}]\ =\ \prod_{i\in [m]} \E[e^{tX_i}]\nonumber
\end{eqnarray}
Using Maclaurin expansion $e^{x} = \sum_{k\geq 0}
\frac{x^k}{k!}$, so for any random variable $X$, 
$\E[e^{tX}] = 
\sum_{k\geq 0}\frac{t^k {\sf E}[ X^k]}{k!}$.
Since $t\geq 0$, \req{eq:indep} follows if we for every $k$ can
prove that $\E[ X^k]\leq \E[\hat{X}^k]$.
We have
\begin{eqnarray} \E[X^k ] & = & \sum_{j_1,\ldots,j_m \mid j_1+ \cdots + j_m = m} {m \choose {j_1,\dots,j_m}} \E[\prod_{i\in[m]} X_i^{j_i}] \label{linearity:eq} \\
& = & \sum_{j_1,\ldots,j_m \mid j_1+ \cdots + j_m = m} {m \choose {j_1,\dots,j_m}} \E[\prod_{i \in [m]| j_i\geq 1} X_i] \label{power:eq} \\
& \le & \sum_{j_1,\ldots,j_m \mid j_1+ \cdots + j_m = m} {m \choose {j_1,\dots,j_m}} \prod_{i \in [m]| j_i\geq 1}   \E[X_i] \label{excp:eq} \\ 
& = & \sum_{j_1,\ldots,j_m \mid j_1+ \cdots + j_m = m} {m \choose {j_1,\dots,j_m}} \prod_{i\in [m]}   \E[X_i^{j_i}] \label{invpower:eq} \\
& = & \sum_{j_1,\ldots,j_m \mid j_1+ \cdots + j_m = m} {m
\choose {j_1,\dots,j_m}} \prod_{i\in [m]}  \E[\hat{X}_{i}^{j_i}] \nonumber \\
  & = & \E[\hat{X}^k ] \label{invlinearity:eq} \ ,
\end{eqnarray}
Here  \req{linearity:eq} and \req{invlinearity:eq} 
follow from linearity of expectation,  \req{power:eq} and \req{invpower:eq} follow from the fact that if $X_i$ is a 
random 0/1 variable then for $j\geq 1$, $X_i^j = X_i$, and \req{excp:eq} follows using (iii$+$)(I).
Hence $\E[e^{tX}] \le
\E[e^{t\hat{X}}]$ as  claimed in \req{eq:indep}. 
\end{proof}
Using Lemma \ref{lem:etS}, we can mimic the standard
proof of Theorem \ref{thm:sri} (I) done with independent variables.
\begin{proof}[ of Theorem~\ref{thm:sri} (I)]
For every $t >0 $
$$
\Pr[X \ge a] = \Pr[e^{tX} \ge e^{ta}] \ .
$$
Using Markov inequality it follows that
\begin{equation} \label{eq:markov}
\Pr[X \ge a] \le \frac{\E[e^{tX}]}{e^{ta}} \ .
\end{equation}
Using Lemma \ref{lem:etS} and arithmetic-geometric
mean inequality,  we get that
\begin{eqnarray*}
\E[e^{tX}] & \leq & \prod_{i\in [n]} \E[e^{t X_i}]\\
& = & \prod_{i\in [n]} \left( 1+ p_i(e^t -1)
\right) \\
& \le & \left( \frac{\sum_{i\in [n]} (1+ p_i(e^t -1))}{m} \right)^m \\
& = & \left( 1 + \frac{\mu}{m} (e^t - 1) \right)^m
\end{eqnarray*}
Substituting this bound into Equation (\ref{eq:markov}) we get that
\begin{equation} \label{eq:y2y}
\Pr[X \ge a] \le \frac{\left( 1 + \frac{\mu}{m} (e^t - 1)
\right)^m}{{e^{t\,a}}} \ .
\end{equation}
Substituting 
$$e^t= \frac{a(n-\mu)}{\mu(n-a)}$$ 
into the right hand side of Equation
(\ref{eq:y2y}) we obtain that
\[
\Pr[X \ge a]  \le   \frac{\left(  1 + \frac{\mu}{m} \left(
\frac{a(n-\mu)}{\mu(n-a)} - 1 \right) \right)^m}{{\left(
\frac{a(n-\mu)}{\mu(n-a)}\right)^{a}}}  
 = \frac{\left(\frac{m-\mu}{m-a}\right)^m}{{\left(
\frac{a(n-\mu)}{\mu(m-a)}\right)^{a}}}   
 = 
 \left(\frac{n-\mu}{m-a}\right)^{m-a}\left(\frac{\mu}{a}\right)^a\ ,
\]
as desired.
\end{proof}
In combination with Lemma \ref{thm:sri-} this completes the proof
of Theorem \ref{thm:sri}. As described in Section \ref{sec:relevance},
this implies that we can use the Chernoff bounds \req{chernoff:upper} and \req{chernoff:lower}
to bound the probability of deviations in our weight estimates.

%\vfill\eject

% \bibliography{paper}
\subsection{Concluding remarks}\label{sec:conclusion}
We presented a general recurrence generating $\varoptk$ schemes for
variance optimal sampling of $k$ items from a set of weighted
items. The schemes provides the best possible average
variance over subsets of any given size. The recurrence
covered previous schemes of Chao and Till\'e \cite{Cha82,Til96}, but
it also allowed us to derive very efficient $\varoptk$ schemes for a streaming
context where the goal is to maintain a reservoir with a sample
of the items seen thus far. We demonstrated the estimate quality experimentally 
against natural competitors such as ppswor and priority sampling.
Finally we showed that the schemes of the recurrence also admits
the kind Chernoff bounds that we normally associate with independent Poisson
sampling for the probability of large deviations.

In this paper, each item is indepdenent. In subsequent work \cite{CDKLT09},
we have considered the unaggregated case where the stream of item have keys, 
and where we are interested in the total weight for each key. Thus, instead
of sampling items, we sample keys. As we sample keys for the reservoir, we
do not know which keys are going to reappear in the future, and for that 
reason, we cannot do a variance optimal sampling of the keys. Yet
we use the variance optimal sampling presented here as a local subroutine.
The result is a heuristic that in experiments outperformed classic
schemes for sampling of unaggregated data like 
sample-and-hold \cite{EV:ATAP02,GM:sigmod98}.

\bibliography{varopt,cycle,replace}
\bibliographystyle{plain}
\appendix
\section{Auxiliary variables}\label{sec:aux-proof}
Continuing from Section \ref{sec:aux} we now consider the case where we for each item are interested in an auxiliary weight $w'_i$.
For these we use the estimate 
\[{\hat w}'_i=w'_i\hat w_i/w_i\]
Let $\VS'=\sum_{i\in[n]}{\hat w}'_i$ 
be the variance on the estimate of the total for the auxiliary variables. 
We want to argue that we expect to do best possible on $\VS'$ using $\varoptk$ that minimizes $\SV$ and $\VS$,
assuming that the $w'_i$ are randomly generated from the $w_i$. Formally
we assume each $w'_i$ is generated as
\[w'_i=x_iw_i\]
where the $x_i$ is drawn independently from the same distribution $\Xi$. 
We consider
expectations $\E_\Xi$ for given random choices of the $x_i$, that is, formally
\[\E_\Xi[\VS']=\E_{x\leftarrow \Xi, i\in[n]}\left[\VS'\,|\,x\right]\]
We want to prove \req{eq:aux}
\[
\E_\Xi[\VS']=\var[\Xi]\SV+\E[\Xi]^2\VS\textnormal,
\]
where $\var[\Xi]=\var_\Xi[x_i]$ and $\E[\Xi]=\E_\Xi[x_i]$ for every $x_i$.
Note that if the $x_i$ are
0/1 variables, then the ${\hat w}'_i$ represent a random subset, including each item independently. This was one of the cases considered in \cite{ST07}. 
However, the general scenario is more like that in statistics where we can think of $w_i$ as a known
approximation of a real weight $w'_i$ which only becomes known if $i$ is actually sampled.
As an example, consider house hold incomes. The $w_i$ could
be an approximation based on street address, but we only find the real incomes $w'_i$ for
those we sample. What \req{eq:aux} states is that if the $w_i'$ are randomly generated
from the $w_i$ by multiplication with independent identically distributed random numbers, then
we minimize the expected variance on the estimate of the real total if our basic
scheme minimizes $\SV$ and $\VS$.

To prove \req{eq:aux}, write
\[\VS'=\SV'+\SCV'\quad\mbox{where}\quad
\SV'=\sum_{i \in [n]} \Var[\hat w'_i]\quad\mbox{and}\quad
\SCV'=\sum_{i,j \in [n]\,|\,i\neq j} \cov[\hat w'_i,\hat w'_j].\]
Here $\var[\hat w'_i]=x_i^2\var[\hat w_i]$ so $\E_\Xi[\var[\hat w'_i]]=\E_\Xi[x_i^2]\var[\hat w_i]$, 
so by linearity of 
expectation, 
\[\E_\Xi[\SV']=\E[\Xi^2]\SV.\]
Similarly, $\cov[\hat w'_i,\hat w'_j]=x_ix_j\cov[\hat w_i,\hat w_j]$ so 
$\E_\Xi[\cov[\hat w'_i,\hat w'_j]]=\E_\Xi[x_i]\E_\Xi[x_j]\cov[\hat w_i,\hat w_j]$, so
by linearity of expectation,
\[\E_\Xi[\SCV']=\E[\Xi]^2\SCV=\E[\Xi]^2(\VS-\SV).\]
Thus 
\[\E_\Xi[\VS']=\E_\Xi[\SV']+\E_\Xi[\SCV']=\E[\Xi^2]\SV+\E[\Xi]^2(\VS-\SV)=\var[\Xi^2]\SV+\E[\Xi]^2\VS\textnormal,\]
as desired.

\section{Bad case for ppswor}\label{sec:ppswor-bad}
We will now provide a generic bad instance for probability proportional to 
size sampling without replacement (ppswor) sampling $k$ out of $n$ items.
Even if ppswor is allowed $k+(\ln k)/2$ samples, it will perform
a factor $\Omega(\log k)$ worse on the average variance for any
subset size $m$ than the optimal scheme with $k$ samples. Since
the optimal scheme has $\VS=0$, it suffices to prove the
statement concerning $\SV$. The negative result is independent
of the ppswor estimator as long as unsampled items get
estimate $0$. The proof of this negative result is only sketched below.

Let $\ell=n-k+1$. The instance has $k-1$ items of size $\ell$ and $\ell$
unit items. The optimal scheme will pick all the large items and one
random unit item. Hence $\SV$ is $\ell(1-1/\ell)\ell<\ell^2$.

Now, with ppswor, there is some probability that a large item
is not picked, and when that happens, it contributes $\ell^2$ to
the variance. We will prove that with the first $k$ ppswor samples,
we waste approximately $\ln k$  samples on unit items, which are
hence missing for the large items, and even if we get half that
many extra samples, the variance contribution from missing
large items is going to be $\Omega(\ell^2\log k)$.

For the analysis, suppose we were going to sample all items
with ppswor. Let $u_i$ be the number of unit items we sample
between the $i-1$st and the $i$th large item. 
Each sample
we get has a probability of almost $(k-i)/(k-i+1)$ of being large.
We say almost because there may be less than $\ell$ remaining
unit items. However, we want to show w.h.p. that close to $\ln k$ unit
items are sampled, so for a contradiction, we can assume that
at least $\ell-\ln k\approx \ell$ unit items remain. As
a result, the expected number of unit items in the interval is 
$(k-i+1)/(k-i)-1=1/(k-i)$.
This means that by the time we get to the $k-\lceil\ln k\rceil$th
large item, the expected number of unit samples
is $\sum_{i=1}^{k-\lceil\ln k\rceil} 1/(k-i)\approx \ln k$.
Since we are adding almost independent random variables each
of which is at most one, we have a sharp concentration, so
by the time we have gotten to the $k-\lceil\ln k\rceil$ large
item, we have approximately $\ln k$ unit samples with high probability.

To get a formal proof using Chernoff bounds, for the number of
unit items between large item $i-1$ and $i$, we can use a pessimistic
0/1 random variable dominated be the above expected number. This
variable is 1 with probability $1/(k-i+1)(1-\ln k/\ell)$ which
is less than the probability that the next item is small, and
now we have independent variables for different rounds.

\end{document}